\documentclass[11pt]{article}

\RequirePackage[OT1]{fontenc}
\RequirePackage[numbers]{natbib}
\RequirePackage[colorlinks,citecolor=black,urlcolor=blue]{hyperref}\usepackage[latin1]{inputenc}
\usepackage{amsmath,amsfonts,amssymb,amsthm,amsopn,amscd}
\usepackage{geometry,stmaryrd}
\usepackage{sidecap}
\usepackage{upgreek}
\usepackage{paralist}

\usepackage{graphicx}
\graphicspath{{figures/}}
\newcommand{\norm}[1]{\left\| #1 \right\|}
\newcommand{\lsup}[1]{\underset{#1\to\infty}{\overline{\lim}}}
\newcommand{\linf}[1]{\underset{#1\to\infty}{\underline{\lim}}}

\theoremstyle{plain}
\newtheorem{theorem}{Theorem}
\newtheorem{corollary}[theorem]{Corollary}

\newtheorem{lemma}[theorem]{Lemma}

\newenvironment{remark}{\par {\noindent \it \sc Remark.} \small \it }

\begin{document}
\title{Population Level Activity in Large Random Neural Networks }

\author{James MacLaurin, Moshe Silverstein, Pedro Vilanova}
\maketitle
\abstract{We determine limiting equations for large asymmetric `spin glass' networks. The initial conditions are not assumed to be independent of the disordered connectivity: one of the main motivations for this is that allows one to understand how the structure of the limiting equations depends on the energy landscape of the random connectivity. The method is to determine the convergence of the double empirical measure (this yields population density equations for the joint distribution of the spins and fields). An additional advantage to utilizing the double empirical measure is that it yields a means of obtaining accurate finite-dimensional approximations to the dynamics.}
\section{Introduction}

This paper concerns the high-dimensional dynamics of asymmetric random neural networks of the form, for $j \in I_N = \lbrace 1,2,\ldots,N \rbrace$.
\begin{equation}\label{eq: Sompolinsky}
dx^j_t = \big(- x^j_t / \tau+ \beta N^{-1/2} \sum_{k=1}^N J^{jk} \lambda(x^k_t) \big)dt + \sigma_t dW^j_t,
\end{equation}
where $\lambda$ is a Lipschitz function, $\tau$ is a constant, and $\lbrace J^{jk} \rbrace_{j,k \in I_N}$ are sampled independently from a centered normal distribution of variance $1$, $\lbrace W^j_t \rbrace_{j\in I_N}$ are Brownian Motions. We study the convergence of the double empirical measure
\begin{align} \label{eq: double empirical measure}
N^{-1}\sum_{j\in I_N} \delta_{(\mathbf{z}^j_{[0,T]}, \mathbf{G}^j_{[0,T]})} ,
\end{align}
where $G^j_t = N^{-1/2} \sum_{k=1}^N J^{jk} \lambda(x^k_t) \big)$. The dynamics of high-dimensional recurrent neural networks have many applications. They have been heavily applied to  neuroscientific problems: many scholars think that they can be used to explain how the brain balances excitation and inhibition \cite{Brunel2000, Ocker2017,Fasoli2019,Landau2018,Cessac2019}. They have been used to study spatially-extended patterns in the brain \cite{Rosenbaum2014,Rosenbaum2017}. Most recently, it has been recognized that they are of fundamental importance to data science \cite{BenArousMei2019,Alaoui2020,Gamarnik2021,Segadlo2022}. For more applications, see the mongraph of Helias and Dahmen \cite{Helias2020} and the recent survey in \cite{Parisi2023}.

There exist limiting `correlation equations' \cite{Crisanti2018,Helias2020} that describe the effective dynamics of high dimensional random neural networks. These constitute delayed integro-differential equations that have proven very difficult to analyze, particularly over short timescales. A related problem is that the correlation equations have only been determined from initial conditions that are independent of the connectivity. This means that they may not be accurate over longer timescales that diverge with $N$. For example, many scholars are interested in understanding the nature of the limiting dynamics after the system attains a particular state (such as, if it enters an `energy well' of specified characteristics, does it escape?). To address this question, one needs to start the dynamics at a particular point in the energy landscape of the connectivity (and therefore the initial condition is disorder-dependent). 

The literature concerning large $N$ limiting equations for random neural networks has a complex history. Sompolinsky, Crisanti and Sommers anticipated that Path Integral methods would yield limiting dynamical equations \cite{Sompolinsky1988} - the derivation was published in a later work \cite{Crisanti2018}. We refer the reader to the excellent discussion in the monograph of Helias and Dahmen \cite{Helias2020}.

Path Integral methods (as practiced by physicists) yield population density equations by determining where the probability measure for the $N$-dimensional system concentrates. In the probability literature, one of the most powerful means of addressing this question is the theory of Large Deviations \cite{Dembo1998}. Large Deviations theory was used to determine spin glass dynamics in the pioneering papers of Ben Arous and Guionnet \cite{BenArous1995,Arous1997,Guionnet1997a}; they obtained the first rigorous results concerning the large $N$ limit of random neural networks. After this work, Grunwald employed Large Deviations theory to obtain correlation / response equations for random neural networks whose spins flip randomly between discrete states \cite{Grunwald1996}. Moynot and Samuelides studied the non-Gaussian case \cite{Moynot2002}. Faugeras and MacLaurin extended the work of Ben Arous and Guionnet to include correlations in the connectivity \cite{Faugeras2015}. Touboul and Cabana determined limiting equations for spatially-extended systems \cite{Cabana2018,Cabana2018a}. Faugeras, Soret and Tanre \cite{Faugeras2019a} determined novel integral equations to describe the state of these systems.

On a related note, correlation-response equations for symmetric random neural networks were first derived by Crisanti, Horner Sommers \cite{Crisanti1993} and Cugliandolo and Kurchan \cite{Cugliandolo1993}. Ben Arous, Dembo and Guionnet \cite{BenArous2006} proved the accuracy of these correlation / response equations for symmetric random neural networks, employing concentration inequalities. 

Broadly-speaking, this paper follows the approach of Ben Arous and Guionnet \cite{BenArous1995}. We employ the theory of Large Deviations to determine the large $N$ limit of the empirical measure. However the main novelties of our approach are:
\begin{itemize}
\item We employ a general class of connectivity-dependent initial conditions. This unsurprisingly yields a different limiting dynamics as $N\to\infty$. Connectivity-dependent initial conditions were employed in the papers of Ben Arous and Guionnet \cite{Arous1997} (who studied dynamics started at the equilibrium distribution, in the high temperature regime) and Dembo and Subag \cite{Dembo2020}.%\cite{BenArous1995}.
\item We study the double empirical measure, that includes information about both the spins and the fields. This has several advantages: it facilitates finite-dimensional approximations to the dynamics that are very accurate, and it facilitates a broader class of disorder-dependent initial condition. For spin-glass dynamics, the Large Deviations of the double empirical measure was determined by Grunwald \cite{Grunwald1996} for jump-Markov systems.
\item We include Replicas (i.e. $M$ copies of the system with the same connectivity, but independent Brownian Motions). This broadens the class of admissible disorder-dependent initial conditions.
\item The function $\lambda$ can be unbounded and the diffusion coefficient $\sigma_t$ can vary in time. The time-varying nature of $\sigma_t$ is essential for studying how periodic environmental noise in the brain shapes the dynamics of random neural networks.
\end{itemize}

\textit{Notation:}
Let $I_N = \lbrace 1,2,\ldots, N \rbrace$ be the set of neuron indices. For any Polish space $\mathcal{X}$, let $\mathcal{P}(\mathcal{X})$ denote all probability measures over $\mathcal{X}$. The space $\mathcal{C}([0,T],\mathbb{R})$ is always endowed with the supremum topology (unless indicated otherwise), i.e.
\[
\norm{x_{[0,T]}} = \sup_{t\in [0,T]} |x_t|
\]
For $\mathbf{y} \in \mathbb{R}^N$, $\norm{\mathbf{y}}$ is the Euclidean norm. 
For any probability measures $\mu$ and $\nu$ over a Polish Space, let $\mathcal{R}(\mu || \nu )$ denotes the relative entropy of measure $\mu$ with respect to $\nu$. For any two measures on the same metric space with metric $d$, $d_W(\cdot,\cdot)$ indicates the Wasserstein distance, i.e.
\begin{equation}
d_W(\mu,\nu) = \inf_{\zeta} \mathbb{E}^{\zeta}\big[d(x,y) \big],
\end{equation}
where the infimum is taken over all $\zeta$ on the product space such that the marginal of the first variable is equal to $\mu$ and the marginal of the second variable is equal to $\nu$. In the particular case that $\mu,\nu \in \mathcal{C}([0,T],\mathbb{R}^M)$, the distance is (unless otherwise indicated) $d(x,y) = \sup_{t\in [0,T]} \sup_{p\in I_M} \big| x^p_t - y^p_t \big|$.

For any $\mu \in \mathcal{P}\big( \mathcal{C}([0,T],\mathbb{R}^M)^2 \big)$, we write $\mu^{(1)} , \mu^{(2)} \in  \mathcal{P}\big( \mathcal{C}([0,T],\mathbb{R}^M) \big)$ to be the marginals over (respectively) the first $M$ variables and last $M$ variables.

\section{Outline of Model and Main Results}

We are going to rigorously determine the limiting dynamics of multiple replicas (with identical connections $\mathbf{J}$, but with independent initial conditions and independent Brownian Motions).  We let the superscript $a$ denote replica $a \in I_M = \lbrace 1,2,\ldots,M \rbrace$, and consider the system
\begin{align}
dz^{a,j}_t =& \big( -  z^{a,j}_t / \tau + G^{a,j}_t  \big)dt + \sigma_t dW^{a,j}_t \text{ where }\label{eq: SDE 1 continuous}\\
G^{a,j}_t =& N^{-1/2}\sum_{k \in I_N} J^{jk}\lambda(z^{a,k}_t).
\end{align}
We assume that $\lambda \in \mathcal{C}^2(\mathbb{R})$: this means in particular that there is a constant $C_{\lambda}$ such that $|\lambda(x) - \lambda(y)| \leq C_{\lambda}|x-y|$. The noise intensity $t \to \sigma_t$ is taken to be continuous and non-random, and such that for constants $\underline{\sigma} $ and $\bar{\sigma}$, 
\begin{equation}
0 < \underline{\sigma} \leq \sigma_t \leq \bar{\sigma}.
\end{equation}
Our major motivation for time-varying diffusivity lies in neuroscience: often synaptic noise exhibits particular rhythms. It has been of major interest how these rhythms shape pattern formation \cite{Brunel2003}.

The connectivities $\lbrace J^{jk}\rbrace$ are taken to be independent centered Gaussian variables, with variance
\[
\mathbb{E}\big[ J^{jk} J^{lm} \big] = \delta(j,l)\delta(k,m) .%+ \mathfrak{s}\delta(j,m)\delta(k,l).
\]
Let $\gamma^N \in \mathcal{P}\big( \mathbb{R}^{N^2} \big)$ be their joint probability law. There are two cases for the initial conditions  $\lbrace z^{a,j}_0\rbrace_{j\in I_N,a\in I_M}$ that are considered in this paper.

%Note that the covariance structure of the Brownian Motions is
%\begin{equation}
%\mathbb{E}\big[ W^{a,j}_s W^{b,k}_t \big] = \delta(a,b) \delta(j,k)\min(s,t).
%\end{equation}
\subsection{Assumptions on the Initial Conditions}

\subsubsection{Case 1: Connectivity-Dependent Initial Conditions }
\label{subsection connectivity dependent initial conditions}
\vspace{.8em}
The probability law of the initial conditions is assumed to be such that for any measurable set $\mathcal{A} \subset \mathbb{R}^{MN}$,
\begin{align}
\mathbb{P}\big( \mathbf{Z}_0 \in \mathcal{A} \big) = \int_{\mathcal{A}} \rho^N_{\mathbf{J}}(\mathbf{x}) d\mathbf{x},
\end{align}
where the probability density $\rho^N_{\mathbf{J}}  :\mathbb{R}^{MN} \to \mathbb{R}^+$ is defined as follows.

Let $\mu^N$ be the uniform Lebesgue Measure over the set $\mathbb{R}^{MN}$. Conditionally on a realization $\mathbf{J}$ of the random connections, let the probability density $\rho^N_{\mathbf{J}}  :\mathbb{R}^{MN} \to \mathbb{R}^+$ be such that for $\delta_N > 0$ (with $\delta_1 = 1$ and $\delta_N$ decreasing to $0$ as $N\to\infty$), there exists some $\kappa \in \mathcal{P}\big( \mathbb{R}^{2M} \big)$ such that
\begin{align}
\rho^N_{\mathbf{J}}(\mathbf{g}) =& \chi\big\lbrace d_W\big( \hat{\mu}^N(\mathbf{z}_0,\mathbf{G}_0) , \kappa\big) \leq \delta_N \big\rbrace / Z^N_{\mathbf{J}} \text{ where }\\
Z^N_{\mathbf{J}} =& \mu^N \big( d_W\big( \hat{\mu}^N(\mathbf{z}_0,\mathbf{G}_0) , \kappa\big) \leq \delta_N \big) \text{ and } \label{eq: Z N J definition} \\
  \hat{\mu}^N(\mathbf{z}_0,\mathbf{G}_0) =& N^{-1}\sum_{j\in I_N} \delta_{\mathbf{z}^j_0 , \mathbf{G}^j_0} \in \mathcal{P}\big(\mathbb{R}^{2M} \big)
\end{align}

Roughly speaking, we need to assume that as $N\to\infty$, the law of $\mathbf{z}^N_0$ behaves like its annealed average. Its assumed that (i) $\kappa$ has a finite second moment in each of its variables, and (ii) we have the bound
\begin{align}
\linf{N} N^{-1}\log  \mathbb{E}[ Z^N_{\mathbf{J}}] > -\infty,  \label{eq: bounded moment 1} 
\end{align}%= \mathcal{R}( \kappa || p ) 
%where we recall that $\mathcal{R}(\cdot || \cdot )$ is the relative entropy, and $p\in \mathcal{P}\big( \mathbb{R}^{2M}\big)$ is the probability measure such that $p$ is the law of random variables $(\mathbf{y}_0, \mathbf{g}_0)$, such that (i) $\mathbf{y}_0$ independent of $\mathbf{g}_0$, (ii) the law of $\mathbf{y}_0$ is $\kappa^{(1)}$, and the law of $\mathbf{g}_0$ is centered and Gaussian, with covariance structure $\mathbb{E}[ g^p_0 g^q_0] = \mathbb{E}[\lambda(y^p_0) \lambda(y^q_0)]$. 
It is also assumed that for any $\epsilon > 0$,
\begin{align}
\lsup{N} N^{-1} \log \mathbb{P}\big( \big| N^{-1}\log Z^N_{\mathbf{J}} - \mathbb{E}[ Z^N_{\mathbf{J}}] \big| \geq \epsilon \big) < 0.\label{eq: bounded moment 2}
\end{align}
%The assumption in \eqref{eq: bounded moment 1} amounts to requiring that $N$ converges to $\infty$ sufficiently slowly that the normalization constant $Z^N_{\mathbf{J}}$ converges to its annealed average. In practice, for these assumptions to be satisfied, one requires that $\delta_N$ converges to $0$ sufficiently slowly - for example $\delta_N = \big(\log N\big)^{-1}$.
Define $\mathfrak{V}_0, \tilde{\mathfrak{V}}_0 \in \mathbb{R}^{M \times M}$ to be the covariance matrix with entries, $p,q\in I_M$,%with entries (using block-matrix indexing), for $p,q\in I_M$ and $u,v \in \lbrace 0,T \rbrace$,
\begin{align}
\mathfrak{V}_{0}^{pq} &= \mathbb{E}^{\kappa}\big[ \lambda(z^{p}_0)\lambda(z^{q}_0)    \big] \label{eq: V 0} \\
\tilde{\mathfrak{V}}_{0}^{pq} &= \mathbb{E}^{\kappa}\big[  z^{p}_0 z^{q}_0   \big].\label{eq: tilde V 0} 
\end{align}
It is also assumed that both $\mathfrak{V}_0$ and $\tilde{\mathfrak{V}}_0$ are invertible.

\subsubsection{Case 2: Connectivity-Independent Initial Conditions }
One can also assume that the initial conditions $(z^{j}_0)_{j\in I_N}$ are (i) independent of the connectivity, and (ii) sampled independently from a $\mathbb{R}^M$-valued probabilistic distribution of bounded variance. This distribution is written as $\hat{\kappa} \in \mathcal{P}(\mathbb{R}^M)$.

\subsection{Main Result}
Our main result is that the empirical measure converges to a fixed point of a mapping $\Phi: \mathcal{U} \to \mathcal{U}$. Here $\mathcal{U} \subset  \mathcal{P}\big( \mathcal{C}([0,T],\mathbb{R}^M)^2 \big)$ is defined in \eqref{eq: U definition}, consisting of (i) a broad class of measures with nice regularity properties, and (ii) such that the empirical measure inhabits $\mathcal{U}$ with unit probability.

For any $\mu \in \mathcal{U}$, in the case of connectivity-dependent initial conditions $\Phi(\mu)$ is specified as follows. It is defined to be the law of Gaussian random variables $\big(z^p_t , G^p_t\big)_{p\in I_M \fatsemi t\in [0,T]}$ such that (i) $(z^p_0,G^p_0)_{p\in I_M}$ are distributed according to $\kappa$, and (ii) conditionally on the initial conditions,  $\big( G^{p}_t \big)_{p\in I_M \fatsemi t\in [0,T]}$ is a  Gaussian system such that $\mathbb{E}[G^{p}_s] = \mathfrak{m}^{p}_{s}(\mu,\mathbf{G}_0)$ and the conditional variance is   
\begin{align}
\mathbb{E}\big[ \big( G^{p}_s  - \mathfrak{m}^{p}_{s}(\mu,\mathbf{G}_0) \big) \big( G^{q}_t -  \mathfrak{m}^{q}_{t}(\mu,\mathbf{G}_0)\big) \; | \; \mathbf{G}_0, \mathbf{z}_0 \big]= \mathfrak{W}^{\mu,pq}_{st}.%\mathbb{E}^{\mu}\big[ \lambda(z^p_t) \lambda(z^q_s) \big].
\end{align}
Here
\begin{align}
\mathfrak{W}^{\mu,pq}_{st} &= \sum_{a,b\in I_M} \mathbb{E}^{\mu}\big[ \lambda(z^{p}_s) \lambda(z^{a}_0) \big] \mathbb{E}^{\mu}\big[ \lambda(z^{q}_t) \lambda(z^{b}_0 \big]  \big(\mathfrak{V}_{\mu,0}^{-1}\big)^{ab}\label{eq: W definition beta 00} \\
\mathfrak{V}_{\mu,0}^{pq} &= \mathbb{E}^{\mu}\big[ \lambda(z^{p}_0)\lambda(z^{q}_0) \big] \\
\mathfrak{m}^{p}_{s}(\mu,\mathbf{g}) &= \sum_{a,b\in I_M} \mathbb{E}^{\mu}\big[ \lambda(z^{p}_s) \lambda(z^{a}_0 \big] \big(\mathfrak{V}_{\mu,0}^{-1}\big)^{ab} g^{b} \label{eq: conditional gaussian 1 beta 00} 
\end{align}
Letting $\big(W^p_{[0,T]}\big)_{p\in I_M}$ be Brownian Motions that are independent of $\mathbf{G}^{\mu}$ , we define $(z^p_t)_{p\in I_M \fatsemi t\in [0,T]}$ to be the strong solution to the stochastic differential equation
\begin{align} \label{eq: R nu definition initial}
dz^p_t = \big( - \tau^{-1} z^p_t + G^{\mu,p}_t \big) dt + \sigma_t dW^p_t.
\end{align}
In the case of connectivity-independent initial conditions, $\Phi$ is defined as follows. One first defines $\big( G^{p}_t \big)_{p\in I_M \fatsemi t\in [0,T]}$ to be a centered Gaussian system such that
\[
\mathbb{E}\big[ G^p_t G^q_s \big] =  \mathbb{E}^{\mu}\big[ \lambda(z^p_t) \lambda(z^q_s) \big].
\]
$(z^p_0)_{p\in I_M}$ is independent of $\big( G^{p}_t \big)_{p\in I_M \fatsemi t\in [0,T]}$ and distributed according to $\hat{\kappa}$. For Brownian Motions $\big(W^p_{[0,T]}\big)_{p\in I_M}$, that are independent of $\mathbf{G}^{\mu}$, $z^p_t$ is the strong solution of \eqref{eq: R nu definition initial}.

%Define $\Phi:  \mathcal{U} \to \mathcal{U}$ to be the following map. For some $\mu \in \mathcal{U}$, write $\Phi(\mu)$ to be the law of the following random variables $(\mathbf{x},\mathbf{h})$. First, it is stipulated that $(\mathbf{x}_0,\mathbf{h}_0)$ have probability law $\kappa$. Second, conditionally on $(\mathbf{x}_0,\mathbf{h}_0)$, the distribution of $( \mathbf{x}_{[0,T]},\mathbf{h}_{[0,T]})$ is given by $S_{\mu,\mathbf{x}_0,\mathbf{g}_0} $. 

\begin{theorem}\label{Theorem Main}
The mapping $\Phi$ is well-defined for all $\mu \in \mathcal{U}$. Furthermore there exists a unique probability measure $\xi \in \mathcal{P}\big( \mathcal{C}([0,T],\mathbb{R}^M)^2 \big)$ such that with unit probability,
\begin{equation}
\lim_{N\to\infty} N^{-1}\sum_{j\in I_N} \delta_{(\mathbf{z}^j_{[0,T]}, \mathbf{G}^j_{[0,T]})} = \xi.
\end{equation}
$\xi$ is the unique measure such that $\Phi(\xi) = \xi$. Furthermore, 
\begin{align}
\xi = \lim_{n\to\infty} \xi^{(n)},
\end{align}
where $\xi^{(n+1)} = \Phi( \xi^{(n)})$ and $\xi^{(1)}$ is any measure in $\mathcal{U}$.

%, where $\Phi$ is a continuous mapping specified in Section \ref{Results Coupled System}.
\end{theorem}
\begin{remark}
This theorem is useful because it also implies a means to efficiently determine the large $N$ limiting equations, through repeated application of the mapping $\Phi$. Because the limiting system is Gaussian, one only needs to solve for its covariance matrix. See the discussion in Helias and Dahmen \cite{Helias2020} for an alternative formulation of the limiting covariance function in terms of a PDE.
\end{remark}
\subsection{An Example System that Satisfies the Conditions of Section \ref{subsection connectivity dependent initial conditions}}

We now outline a general system that satisfies the conditions of Section \ref{subsection connectivity dependent initial conditions}. Suppose that $\lambda : \mathbb{R} \to \mathbb{R}$ is an odd function. Define $\iota: \mathcal{P}\big( \mathbb{R}^{2}  \times \mathbb{R}^{2}  \big) \to \mathbb{R}^+$ to be
\begin{align}
\iota(\mu) = \mathcal{R}( \mu || p).
\end{align}
Here $p \in \mathcal{P}\big( \mathbb{R}^{2}  \times \mathbb{R}^{2}  \big)$ is centered and Gaussian, the law of random variables $(z^1_0, g^1_0, z^2_0, g^2_0)$ that are such that $\mathbb{E}[(z^p_0)^2] = 1$, $\mathbb{E}[z^p_0 g^q_0] = 0$  and $\mathbb{E}[g^p_0 g^q_0] = \mathbb{E}[ \lambda(z^p_0)\lambda(z^q_0)]$. Let $\Xi: \mathcal{C}\big( \mathbb{R}^{2} \big) \to \mathbb{R}$ be bounded.  
\begin{lemma}
Suppose that there is a unique $\eta \in\mathcal{P}\big( \mathbb{R}^{2} \times \mathbb{R}^2\big)$ such that
\begin{align}
\sup_{\mu \in  \mathcal{P}\big( \mathbb{R}^{2}  \times \mathbb{R}^{2}  \big)}  \big\lbrace \mathbb{E}^{\mu}[ \Xi(z^1,g^1)  ]- \iota(\mu) \big\rbrace = \mathbb{E}^{\eta}[ \Xi(z^1,g^1) ] - \iota(\eta).
\end{align}
Suppose also that
\begin{multline}
\sup \big\lbrace \mathbb{E}^{\mu}[ \Xi(z^1,g^1) + \Xi(z^2,g^2) ]- \iota(\mu) \; : \; \mu \in \mathcal{P}\big( \mathbb{R}^{2}  \times \mathbb{R}^{2}  \big) \text{ and } \mu^{(1)} = \eta^{(1)} \text{ and }\mu^{(2)} = \eta^{(1)} \big\rbrace \\= 2 \big\lbrace \mathbb{E}^{\eta}[ \Xi(z^1,g^1) ] - \iota(\eta) \big\rbrace .%\mathbb{E}^{\zeta}[ \Xi(z^1,g^1) + \Xi(z^2,g^2) ] - \iota(\zeta).
\end{multline}
Then the conditions of Section \ref{subsection connectivity dependent initial conditions} are satisfied, substituting $M=1$, $\kappa = \eta^{(1)}$ and $\delta_N = (\log N)^{-1}$.
\end{lemma}
\begin{proof}
We show that
\begin{align}\label{eq: double moment Z N J}
\lsup{N} N^{-1} \big\lbrace \log \mathbb{E}^{\gamma}\big[ (Z^N_{\mathbf{J}})^2 \big] - 2 \log \mathbb{E}^{\gamma}\big[ Z^N_{\mathbf{J}} \big] \big\rbrace= 0.
\end{align}
It is immediate from \eqref{eq: double moment Z N J} that \eqref{eq: bounded moment 2} must be satisfied. Let $\lbrace y^{p,j}_0 \rbrace_{j\in I_N \fatsemi 1\leq p \leq 2}$ be iid $\mathcal{N}(0,1)$ random variables. Form the empirical measure $\hat{\mu}^N_0 := N^{-1}\sum_{j\in I_N} \delta_{\mathbf{y}^j, \mathbf{G}^j_0} \in \mathcal{P}\big(\mathbb{R}^{4} \big)$.

Define
\begin{align*}
\mathcal{A}_N = \chi\bigg\lbrace d_W\big( \eta^{(1)}, N^{-1}\sum_{j\in I_N} \delta_{y^{1,j}, g^{1,j}}  \big) \leq (\log N)^{-1} \; , \;  d_W\big( \eta^{(2)}, N^{-1}\sum_{j\in I_N} \delta_{y^{2,j}, g^{2,j}}  \big) \leq (\log N)^{-1} \bigg\rbrace .
\end{align*}
We will first demonstrate that
\begin{multline}
\lim_{N\to\infty} N^{-1}  \log \mathbb{E}\big[ \mathcal{A}_N \exp\big( N \mathbb{E}^{\hat{\mu}^N}[ \Xi(z^1,g^1) + \Xi(z^2,g^2)] \big) \big] = \\
\sup \big\lbrace \mathbb{E}^{\mu}[ \Xi(z^1,g^1) + \Xi(z^2,g^2) ]- \iota(\mu) \; : \; \mu \in \mathcal{P}\big( \mathbb{R}^{2}  \times \mathbb{R}^{2}  \big) \text{ and } \mu^{(1)} = \eta^{(1)} \text{ and }\mu^{(2)} = \eta^{(1)} \big\rbrace  .\label{eq: first to demonstrate}
\end{multline}% 2 \big\lbrace \mathbb{E}^{\eta}[ \Xi(z^1,g^1) ] - \iota(\eta) \big\rbrace
It is straightforward to prove that for any $a > 0$ there exists a compact subset $ \mathfrak{U}_a \subset \mathcal{P}\big(\mathbb{R}^{4} \big)$ such that
\begin{align}
\lsup{N} N^{-1}\log \mathbb{P}\big( \hat{\mu}^N_0  \notin \mathfrak{U}_a \big) \leq - a.
\end{align}
A simple Large Deviations estimate yields that, for any $\mu \in \mathcal{P}\big(\mathbb{R}^{4} \big)$, and $\epsilon \ll 1$,
\begin{align}
N^{-1}\log \mathbb{P}\big( \hat{\mu}^N_0  \in B_{\epsilon}(\mu) \big) = - \iota(\mu) + O(\sqrt{\epsilon}). \label{eq: Large Deviations application}
\end{align}
(In fact this could also be proved using Corollary \ref{Corollary Uncoupled Large Deviation Principle Unconditioned}). Our choice that $\delta_N$ goes to zero sufficiently slowly (i.e. $\delta_N = (\log N)^{-1}$) ensures that the rate function in \eqref{eq: Large Deviations application} dominates the asymptotic estimate of the probability as $N\to\infty$. Thus discretizing $\mathfrak{U}_a$ into $\epsilon$ balls, and then taking $\epsilon \to 0^+$, one obtains \eqref{eq: first to demonstrate}.

Now \eqref{eq: first to demonstrate} also implies that
\begin{align}\label{eq: iota inequality}
\lim_{N\to\infty} N^{-1}   \log \mathbb{E}\big[ \mathcal{A}_N  \big] = -2\iota(\eta).
\end{align}
This holds because since $\mu \to \mathbb{E}^{\mu}[\Xi(y^1,g^1 + \Xi(y^2,g^2)]$ is continuous, it becomes almost constant as $N\to\infty$, as the radius of $\mathcal{A}_N$ shrinks. \eqref{eq: iota inequality} in turn implies \eqref{eq: double moment Z N J}.
%For the definition of \eqref{eq: Z N J definition}
\end{proof}

\section{Proof Outline}
The main goal of this paper is to prove Theorem \ref{Theorem Main} employing the theory of Large Deviations \cite{Dembo1998}. The method - similarly to the original work by Ben Arous and Guionnet \cite{BenArous1995} - is to (i) prove a Large Deviations Principle for the uncoupled system, and then (ii) perform an exponential change-of-measure using Girsanov's Theorem to obtain the Large Deviations Principle for the coupled system, before (iii) proving that the rate function has a unique zero. 

The three main differences between this paper and the early papers of Ben Arous and Guionnet is that we (i) study the convergence of the double empirical measure \eqref{eq: double empirical measure} (whereas Ben Arous and Guionnet study the convergence of the annealed empirical measure in their earlier papers \cite{BenArous1995}) (in the later works \cite{Guionnet1997,BenArous2006} quenched asymptotics are determined) (ii) we employ disorder-dependent initial conditions and (iii) we employ replicas. 

Our main focus is on proving Case 1 (i.e. the connectivity-dependent initial conditions). The proofs are broadly similar, however Case 1 is more difficult because it requires a uniform Large Deviations Principle for the conditioned probability laws.
\subsection{Large Deviations of the Uncoupled System}
We start by noting a Large Deviation Principle for the uncoupled system. Because we are employing general disorder-dependent initial conditions, we must determine a Large Deviations Principle for the conditioned probability law. To this end, we must first define the set $\mathcal{Y}^N$ of all `valid initial conditions' (basically the set of all initial points such that the empirical measure at time $0$ is close to its limit). More precisely, we define $\mathcal{Y}^N \subset \mathbb{R}^{NM} \times \mathbb{R}^{NM}$ to be such that
\begin{align}
\mathcal{Y}^N= \big\lbrace (\mathbf{z}_0,\mathbf{g}_0) \; : \;  d_W\big( \hat{\mu}^N(\mathbf{z}_0,\mathbf{g}_0) , \kappa\big) \leq \tilde{\delta}_N \big\rbrace , \label{eq: Y N set}
\end{align}
where $\tilde{\delta}_N = \max\big\lbrace \delta_N ,( \log N)^{-1} \big\rbrace$ and $ \hat{\mu}^N(\mathbf{z}_0,\mathbf{g}_0) = N^{-1}\sum_{j\in I_N}\delta_{\mathbf{z}^j_0, \mathbf{g}^j_0} \in \mathcal{P}(\mathbb{R}^{2M} )$. Define the uncoupled dynamics,
\begin{align}
y^{p,j}_t = z^{p,j}_0 + \int_0^t \sigma_s dW_s^{p,j}, \label{eq: y dynamics}
\end{align}
and let $P^N_{\mathbf{z}_0} \in \mathcal{P}\big( \mathcal{C}([0,T],\mathbb{R}^M)^N \big)$ be the law of $\lbrace y^j_{[0,T]} \rbrace_{j\in I_N}$, conditioned on the values at time $0$. Write 
\begin{align}
\tilde{G}^{p,j}_t = N^{-1/2}\sum_{k\in I_N} J^{jk} \lambda(y^{p,k}_t).
\end{align}
Define $\gamma^N_{\mathbf{g}_0} \in \mathcal{P}\big( \mathbb{R}^{N^2} \big)$ to be the law of the connections $\lbrace J^{jk} \rbrace_{j,k \in I_N}$, conditioned on the event
\begin{align}
\big\lbrace \tilde{G}^{p,j}_0 = g^{p,j}_0 \text{ for each }j\in I_N, \; p\in I_M \big\rbrace . \label{eq: conditioned event}
\end{align}
We note that $\gamma^N_{\mathbf{g}_0}$ is Gaussian, but no longer of zero mean (in general). The mean of $\tilde{G}^{p,j}_t$ is a function of the empirical measure and $\mathbf{g}^j_0$ (explicit formulae are outlined further below). Let $Q^N_{\mathbf{z}_0,\mathbf{g}_0} \in \mathcal{P}\big( \mathbb{R}^{N^2} \times \mathcal{C}([0,T],\mathbb{R}^M)^N \big)$ be the joint law of the uncoupled system, i.e.
\begin{equation}
Q^N_{\mathbf{z}_0,\mathbf{g}_0} =  \gamma^N_{\mathbf{g}_0} \otimes P^N_{\mathbf{z}_0} .
\end{equation}
The first main result is a uniform Large Deviations Principle for the conditioned system.
\begin{theorem}\label{Theorem Uncoupled Large Deviation Principle Conditional Time 0 0}
Let $\mathcal{A},\mathcal{O}\in \mathcal{B}\big( \mathcal{P}\big( \mathcal{C}([0,T],\mathbb{R}^M)^2\big) \big)$, such that $\mathcal{O}$ is open and $\mathcal{A}$ closed. Then
\begin{align}
 \lsup{N}\sup_{(\mathbf{z}_0,\mathbf{g}_0) \in \mathcal{Y}^N} N^{-1}\log Q^N_{\mathbf{z}_0,\mathbf{g}_0}\big( \hat{\mu}^N(\mathbf{y}_{[0,T]}, \tilde{\mathbf{G}}_{[0,T]}) \in \mathcal{A}\big) &\leq -\inf_{\mu \in \mathcal{A}}\mathcal{I}(\mu)\label{LDP upper bound} \\
\linf{N}\inf_{(\mathbf{z}_0,\mathbf{g}_0) \in \mathcal{Y}^N} N^{-1}\log Q^N_{\mathbf{z}_0,\mathbf{g}_0} \big( \tilde{\mu}^N(\mathbf{y}_{[0,T]}, \tilde{\mathbf{G}}_{[0,T]}) \in \mathcal{O}\big) &\geq -\inf_{\mu \in \mathcal{O}}\mathcal{I}(\mu) . \label{LDP lower bound}
\end{align}
Here the rate function $\mathcal{I}:  \mathcal{P}\big( \mathcal{C}([0,T],\mathbb{R}^M)^2\big)  \to \mathbb{R}^+$ is such that $\mathcal{I}(\mu)  = \infty$ if $\mu \notin \mathcal{U}$, else otherwise
\begin{align}
\mathcal{I}(\mu) = \tilde{I}_{\mu}(\mu),
\end{align}
where, $\tilde{I}_{\mu}(\mu)$ is defined in \eqref{eq: tilde I rate function}.  Furthermore $\mathcal{I}$ is lower-semi-continuous and has compact level sets.
\end{theorem}
We will prove Theorem \ref{Theorem Uncoupled Large Deviation Principle Conditional Time 0 0} by locally freezing the dependence of the fields $\lbrace \tilde{G}^{p,j}_t \rbrace$ on the empirical measure. In order that we may do this, we must first define a regular subset $\mathcal{Q}_{\mathfrak{a}}$ (for a positive integer $\mathfrak{a} \gg 1$) which is such that (i) the empirical measure $\hat{\mu}^N(\mathbf{y}) = N^{-1}\sum_{j\in I_N} \delta_{\mathbf{y}^j_{[0,T]}} \in \mathcal{P}\big( \mathcal{C}([0,T],\mathbb{R}^M) \big)$ inhabits with high probability, and (ii) there exist uniform bounds on the fluctuations in time. To this end, writing $\mathcal{K}_{\mathfrak{a}}$ to be the compact set specified in Lemma \ref{Lemma compact K L}, define the set
\begin{multline} \label{eq: set Q L}
\mathcal{Q}_{\mathfrak{a}} = \bigg\lbrace \mu \in \mathcal{P}\big( \mathcal{C}([0,T], \mathbb{R}^M) \big) \; : \; \mu \in \mathcal{K}_{\mathfrak{a}} \text{ and }\sup_{p\in I_M} \mathbb{E}^{\mu}[ \sup_{t\in [0,T]}(y^p_t)^2 \big] \leq \mathfrak{a} \text{ and }\\ \text{ For all integers }m \geq \mathfrak{a} \text{ it holds that } \sup_{0\leq i \leq m}\mathbb{E}^{\mu}\big[ \sup_{p \in I_M} (w^p_{t^{(m)}_{i+1}} - w^p_{t^{(m)}_i})^2\big] \leq \Delta_m^{1/4}
\bigg\rbrace
\end{multline}
where $\Delta_m = T / m$ and $t^{(m)}_i = iT/m$. Write
\begin{align}
\mathfrak{Q} = \bigcup_{\mathfrak{a}\geq 1}\mathfrak{Q}_{\mathfrak{a}}.
\end{align}
\begin{lemma} \label{Lemma y exponential tightness}
For any $L >0$, there exists $a> 0$ such that
\begin{align}
\lsup{N} N^{-1}\log \mathbb{P}\big( \hat{\mu}^N(\mathbf{y}) \notin  \mathcal{Q}_a \big) \leq - L
\end{align}
\end{lemma}
The above lemma is proved in the Appendix. Next, for any $\nu \in \mathcal{Q} $, we define a centered Gaussian law $\beta_{\nu} \in  \mathcal{P}\big( \mathcal{C}([0,T],\mathbb{R})^M \big)$ as follows. We stipulate that $\beta_{\nu}$ is the law of Gaussian random variables $\lbrace G^{\nu,p}_t \rbrace_{t\in [0,T], p\in I_M}$ with covariance structure
\begin{align}
\mathbb{E}^{\beta_{\nu}}\big[ G^{\nu,p}_sG^{\nu,q}_t \big] = \mathbb{E}^{\nu}\big[ \lambda(x^p_s) \lambda(x^q_t) \big] \label{eq: beta nu covariance definition}
\end{align}
This definition will be useful because for any $j \in I_N$, the law of $\tilde{G}^j_{[0,T]} $ under $\gamma^N$ is $\beta_{\hat{\mu}^N(\mathbf{y})}$. In the following Lemma we collect some regularity estimates for the Gaussian Law $\beta_{\nu}$.
\begin{lemma} \label{Lemma beta nu well defined}
(i) $\beta_{\nu}$ is a well-defined Gaussian probability law. (ii) Furthermore, the map $t\to G^{\nu,p}_t$ is `uniformly continuous' for all measures in $\mathcal{U}_a$, in the following sense. For any $a > 0$, and any $\epsilon > 0$, there exists $\delta(a,\epsilon)$ such that for all $\nu \in \mathcal{U}_a$,
\begin{align}
\sup_{\nu \in \mathcal{U}_a} \sup_{p\in I_M} \mathbb{E}^{\beta_\nu} \bigg[ \sup_{s,t \in [0,T] \fatsemi |s-t| \leq \delta(a,\epsilon)} \big| G^{\nu,p}_s - G^{\nu,p}_t \big| \bigg] \leq \epsilon \label{eq: uniform continuity U a G}
\end{align}
\end{lemma}
In order that we may make sense of the disorder-dependent initial condition, we also require an understanding of the distribution of $\beta_{\nu}$, conditioned on the value of $\tilde{G}_0$. To this end, for any $\nu \in \mathcal{U}$ and any $\mathbf{g} \in \mathbb{R}^M$, let $\beta_{\nu,\mathbf{g}} \in \mathcal{P}\big( \mathcal{C}([0,T],\mathbb{R}^M) \big)$ be the probability law $\beta_{\nu}$ conditioned on the event $\mathbf{G}^{\nu}_0 = \mathbf{g}$. Standard Gaussian identities \cite{Lindgren2013} imply that $\beta_{\nu,\mathbf{g}}$ is also Gaussian, with the following mean and variance,
\begin{align}
\mathbb{E}^{\beta_{\nu,\mathbf{g}}} \big[ G^{\nu,p}_s \big] =& \mathfrak{m}^{p}_{s}(\nu,\mathbf{g}) \label{eq: mps average beta} \\
\mathbb{E}^{\beta_{\nu,\mathbf{g}}} \big[ \big(G^{\nu,p}_s -  \mathfrak{m}^{p}_{s}(\nu,\mathbf{g}) \big) \big(G^{\nu,q}_t -  \mathfrak{m}^{q}_{t}(\nu,\mathbf{g} \big)  \big] =&  \mathfrak{W}^{\nu,pq}_{st} ,\label{eq: Vps average beta} 
\end{align}
where%Write $\mathfrak{v}^{pq}_{st}: \mathcal{P}\big( \mathcal{C}([0,T],\mathbb{R}^M) \big) \times \mathbb{R}^M \to \mathbb{R}$,
\begin{align}
\mathfrak{W}^{\mu,pq}_{st} &= \sum_{a,b\in I_M} \mathbb{E}^{\mu}\big[ \lambda(z^{p}_s) \lambda(z^{a}_0) \big] \mathbb{E}^{\mu}\big[ \lambda(z^{q}_t) \lambda(z^{b}_0 \big]  \big(\mathfrak{V}_{\mu,0}^{-1}\big)^{ab}\label{eq: W definition beta} \\
\mathfrak{m}^{p}_{s}(\mu,\mathbf{g}) &= \sum_{a,b\in I_M} \mathbb{E}^{\mu}\big[ \lambda(z^{p}_s) \lambda(z^{a}_0 \big] \big(\mathfrak{V}_{\mu,0}^{-1}\big)^{ab} g^{b} \label{eq: conditional gaussian 1 beta} \\
\mathfrak{V}_{\mu,0}^{pq} &= \mathbb{E}^{\mu}\big[ \lambda(z^{p}_0)\lambda(z^{q}_0)    \big] .
\end{align}
\begin{corollary}\label{Lemma CompactNess Fields}
(i) For any $a > 0$, and any $\epsilon > 0$, there exists $\tilde{\delta}(a,\epsilon)$ such that for all $\nu \in \mathcal{U}_a$, and all $\mathbf{g} \in \mathbb{R}^M$,
\begin{align}
\sup_{\nu \in \mathcal{U}_a} \sup_{p\in I_M} \mathbb{E}^{\beta_{\nu, \mathbf{g}}} \bigg[ \sup_{s,t \in [0,T] \fatsemi |s-t| \leq \tilde{\delta}(a,\epsilon)} \big| G^{\nu,p}_s - \mathfrak{m}^{p}_{s}(\nu,\mathbf{g})- G^{\nu,p}_t + \mathfrak{m}^{p}_{t}(\nu,\mathbf{g})\big| \bigg] \leq \epsilon \label{eq: uniform continuity U a G conditioned}
\end{align}
(ii) For any $\epsilon , a > 0$, there exists a compact set $\mathcal{C}_{\epsilon,a} \subset \mathcal{C}([0,T],\mathbb{R}^M)$ such that for all $ \nu \in \mathcal{Q}_a$,
\begin{align} \label{eq: corolalry uniform 1}
\beta_{\nu}(\mathcal{C}_{\epsilon,a}) \geq 1- \epsilon,
\end{align}
and for all $\mathbf{g}_0 \in \mathbb{R}^M$ such that $\norm{\mathbf{g}_0} \leq a$,
\begin{align}\label{eq: corolalry uniform 2}
\beta_{\nu,\mathbf{g}_0}(\mathcal{C}_{\epsilon,a}) \geq 1- \epsilon.
\end{align}
(iii) For $1\leq j \leq N$ and any $\mathbf{g}_0 \in \mathbb{R}^M$, the law of $\tilde{G}^j_{[0,T]} $ under $\gamma^N_{\mathbf{g}_0}$ is identical to $\beta_{\hat{\mu}^N(\mathbf{y}) , \mathbf{g}_0}$.
\end{corollary}

\subsection{Exponential Tightness}
To prove a Large Deviation Principle, one requires that the empirical measure inhabits a compact set with arbitrarily high probability. For any $\mathbf{y} \in \mathcal{C}([0,T],\mathbb{R}^M)^N$, write $\tilde{\gamma}^N_{\mathbf{y}} \in \mathcal{P}\big(  \mathcal{C}([0,T],\mathbb{R}^M)^N \big)$ to be the law of the random variables $(\tilde{G}^{p,j}_t)_{j\in I_N \fatsemi p\in I_M \fatsemi t\in [0,T]}$, and write $\tilde{\gamma}^N_{\mathbf{y},\mathbf{g}_0}$ to be $\tilde{\gamma}^N_{\mathbf{y}}$ conditioned on the event in \eqref{eq: conditioned event}.

The following lemmas are needed for this proof.
\begin{lemma} \label{Lemma Exponential Tightness mu G}
For any $L > 0$, there exists a compact set $\tilde{\mathcal{C}}_L \subset \mathcal{P}\big( \mathcal{C}([0,T],\mathbb{R}^M) \big)$ such that the following holds. For any $N\geq 1$, and any $\lbrace \mathbf{y}^j_{[0,T]} \rbrace_{j\in I_N}$ such that $\hat{\mu}^N(\mathbf{y}) \in \mathcal{Q}_L$, 
%There exists a universal constant such that for any $\Delta > 0$, $N \geq 1$ and $t \leq T-\Delta$,
\begin{align}
N^{-1}\log \tilde{\gamma}^N_{\mathbf{y}}\big( \hat{\mu}^N(\tilde{\mathbf{G}}) \notin \tilde{\mathcal{C}}_L  \big) \leq -L.
\end{align}
Also, as long as $(\mathbf{y}_0,\mathbf{g}_0) \in \mathcal{Y}^N$,
\begin{align}
\lsup{N} \sup_{(\mathbf{y}_0,\mathbf{g}_0) \in \mathcal{Y}^N}N^{-1}\log \tilde{\gamma}^N_{\mathbf{y}, \mathbf{g}_0}\big( \hat{\mu}^N(\tilde{\mathbf{G}}) \notin \tilde{\mathcal{C}}_L  \big) \leq -L. \label{Lemma Exp Tightness G}
\end{align}
\end{lemma}

For $\mu \in  \mathcal{P}\big( \mathcal{C}([0,T],\mathbb{R})^M \times  \mathcal{C}([0,T],\mathbb{R})^M \big)$, write $\mu^{(1)} \in  \mathcal{P}\big( \mathcal{C}([0,T],\mathbb{R})^M \big)$ to be the marginal of $\mu$ over its first $M$ variables, and $\mu^{(2)}$ to be the marginal of $\mu$ over its last $M$ variables. Next, define the set
\begin{multline}
\mathcal{U}_a = \bigg\lbrace \mu \in \mathcal{P}\big( \mathcal{C}([0,T],\mathbb{R})^M \times  \mathcal{C}([0,T],\mathbb{R})^M \big) \; : \; \mu^{(1)} \in \mathcal{Q}_{a}, \; \mu^{(2)} \in \tilde{\mathcal{C}}_a \text{ and } \\  \sup_{t\in [0,T]}\sup_{p\in I_M} \mathbb{E}^{\mu}[(G^p_t)^2] \leq C_{\lambda}^2 a, 
\;\text{for all }0\leq s,t \leq T,\; \sup_{p\in I_M} \mathbb{E}^{\mu}[(G^p_t - G^p_s)^2] \leq a C_{\lambda}^2 |t-s|^{1/2}   \bigg\rbrace , \label{eq: nu bounded moment} %\; , \; \sup_{t\in [0,T]} \sup_{p\in I_M}\mathbb{E}^{\mu}[(G^p_t)^2] \leq a 
\end{multline}
and let
\begin{align}
\mathcal{U} = \bigcup_{a\geq 0} \mathcal{U}_a. \label{eq: U definition}
\end{align}
It follows immediately from the above definition that $d_W(\mu,\nu) < \infty$ for any $\mu,\nu \in \mathcal{U}$. We can now prove an `exponential tightness' result.
\begin{lemma} \label{Lemma Exponential Tightness}
For any $a \geq 0$, $\mathcal{U}_a$ is compact. For any $L > 0$, there exists $a > 0$ such that
\begin{align}
 \lsup{N}\sup_{(\mathbf{z}_0,\mathbf{g}_0) \in \mathcal{Y}^N} N^{-1}\log Q^N_{\mathbf{z}_0,\mathbf{g}_0}\big( \hat{\mu}^N \notin \mathcal{U}_a \big) \leq - L.
\end{align}
\end{lemma}
\begin{proof}
Since the sets $\mathcal{Q}_a$ and $\tilde{\mathcal{C}}_a$ are compact, this follows almost immediately from Lemmas \ref{Lemma y exponential tightness} and \ref{Lemma Exponential Tightness mu G}.
\end{proof}

\subsection{The Coupled System (with connectivity-dependent initial conditions)} \label{Results Coupled System}

Define $\mathcal{J}:  \mathcal{P}\big( \mathcal{C}([0,T],\mathbb{R}^M)^2\big) \to \mathbb{R}$ to be such that $\mathcal{J}(\mu) = \infty$ if $\mu \notin \mathcal{U}$, or if the marginal of $\mu$ at time $0$ is not $\kappa$. Else otherwise, for any $\mu \in  \mathcal{P}\big( \mathcal{C}([0,T],\mathbb{R}^M)^2\big)$ and $\mathbf{z}_0,\mathbf{g}_0 \in \mathbb{R}^M$, writing $\mu_{\mathbf{z}_0,\mathbf{g}_0}$ for $\mu$ conditioned on the values of its variables at time $0$, define
\begin{align}
\mathcal{J}(\mu) = \mathbb{E}^{\kappa}\bigg[ \mathcal{R}\big(\mu_{\mathbf{z}_0,\mathbf{g}_0} || S_{\mu,\mathbf{z}_0,\mathbf{g}_0} \big) \bigg] \label{eq: J rate function}
\end{align}
Here $S_{\mu,\mathbf{z}_0,\mathbf{g}_0} \in \mathcal{P}\big( \mathcal{C}([0,T],\mathbb{R}^M)^2 \big)$ is defined to be the probability law of $(\mathbf{z},\mathbf{G}^{\mu})$, where $\mathbf{G}^{\mu}$ is distributed according to $\beta_{\mu, \mathbf{g}_0}$, and for Brownian Motions $\big(W^p_{[0,T]}\big)_{p\in I_M}$ that are independent of $\mathbf{G}^{\nu}$ 
\begin{align} \label{eq: R nu definition}
dz^p_t = \big( - \tau^{-1} z^p_t + G^{\mu,p}_t \big) dt + \sigma_t dW^p_t
\end{align}

Define $\Phi:  \mathcal{U} \to \mathcal{U}$ to be the following map. For some $\mu \in \mathcal{U}$, write $\Phi(\mu)$ to be the law of the following random variables $(\mathbf{x},\mathbf{h})$. First, it is stipulated that $(\mathbf{x}_0,\mathbf{h}_0)$ have probability law $\kappa$. Second, conditionally on $(\mathbf{x}_0,\mathbf{h}_0)$, the distribution of $( \mathbf{x}_{[0,T]},\mathbf{h}_{[0,T]})$ is given by $S_{\mu,\mathbf{x}_0,\mathbf{g}_0} $. 

\begin{lemma}\label{Lemma unique minimizer rate function}
The probability law $S_{\mu,\mathbf{z}_0,\mathbf{g}_0}$ is well-defined. Furthermore, there exists a unique zero $\xi$ of the rate function $\mathcal{J}$. $\xi$ is the unique measure in $\mathcal{U}$ such that $\Phi(\xi)=\xi$. 
\end{lemma}
\begin{proof}
We have already proved in Lemma \ref{Lemma beta nu well defined} that $\big( G^{\mu,p}_{[0,T]} \big)_{p\in I_M}$ is well-defined.
It is straightforward to check that for any path $\big( G^{\mu,p}_{[0,T]} \big)_{p\in I_M}$, there exists a unique strong solution to the stochastic differential equation \eqref{eq: R nu definition}. Thus the probability law is well-defined.

It is well-known that the Relative Entropy is zero if and only if its two arguments are identical \cite{Budhiraja2019}. Thus, from the form of $\mathcal{J}$ in \eqref{eq: J rate function}, any zero must be a fixed point of $\Phi$. Furthermore, there must exist at least one zero of the rate function (if not, the total probability mass could not be one as $N\to\infty$). The uniqueness of the zero is proved in Lemma \ref{Lemma at most one Fixed Point}.
\end{proof}
\begin{theorem}\label{Theorem Coupled Large Deviation Principle}
For any $\epsilon > 0$,
\begin{align}
\lsup{N} N^{-1} \log \mathbb{P}\big( d_W(\hat{\mu}^N(\mathbf{z},\mathbf{G}) , \xi) \geq \epsilon \big) < 0.
\end{align}
Thus with unit probability,
\begin{equation}
\lim_{N\to\infty}\hat{\mu}^N(\mathbf{z},\mathbf{G})  = \xi .
\end{equation}
Furthermore, 
\begin{align}
\xi = \lim_{n\to\infty} \xi^{(n)}, \label{eq: convergence xi n approximations}
\end{align}
where $\xi^{(n+1)} = \Phi( \xi^{(n)})$ and $\xi^{(1)}$ is any measure in $\mathcal{U}$.
\end{theorem}
\subsection{Connectivity-Independent Initial Conditions}
The above reasoning can be adapted to prove a Large Deviation Principle for the unconditioned system. This is needed for proving the main theorem for Case 2 (connectivity-independent initial conditions). Write $Q^N = \gamma^N \otimes P^N$ to be the law of the random variables $(\mathbf{y},\mathbf{G})$ (with no conditioning), and for any $\nu \in \mathfrak{Q}$, define $S_{\nu} \in \mathcal{P}\big( \mathcal{C}([0,T],\mathbb{R}^M)^2 \big)$ to be $S_{\nu} = P \otimes \beta_{\nu}$. In the following  corollary to Theorem \ref{Theorem Uncoupled Large Deviation Principle Conditional Time 0 0}, we prove the Large Deviation Principle for the unconditioned and uncoupled system.
\begin{corollary}\label{Corollary Uncoupled Large Deviation Principle Unconditioned}
Let $\mathcal{A},\mathcal{O}\in \mathcal{B}\big( \mathcal{P}\big( \mathcal{C}([0,T],\mathbb{R}^M)^2\big) \big)$, such that $\mathcal{O}$ is open and $\mathcal{A}$ closed. Then
\begin{align}
 \lsup{N} N^{-1}\log Q^N\big( \hat{\mu}^N(\mathbf{y}_{[0,T]}, \tilde{\mathbf{G}}_{[0,T]}) \in \mathcal{A}\big) &\leq -\inf_{\mu \in \mathcal{A}} \mathcal{R}\big(\mu || S_{\mu^{(1)}}\big)\label{LDP upper bound corollary}  \\
\linf{N} N^{-1}\log Q^N  \big( \hat{\mu}^N(\mathbf{y}_{[0,T]}, \tilde{\mathbf{G}}_{[0,T]}) \in \mathcal{O}\big) &\geq -\inf_{\mu \in \mathcal{O}}\mathcal{R}\big(\mu || S_{\mu^{(1)}}\big). \label{LDP lower bound corollary}
\end{align}
Here the rate function $\mu \to \mathcal{R}\big(\mu || S_{\mu^{(1)}}\big)$ is lower semi-continuous and has compact level sets. %:  \mathcal{P}\big( \mathcal{C}([0,T],\mathbb{R}^M)^2\big)  \to \mathbb{R}^+$ is such that $\mathcal{I}(\mu)  = \infty$ if $\mu \notin \mathcal{U}$, else otherwise
%\begin{align}
%\mathcal{I}(\mu) = \tilde{I}_{\mu}(\mu),
%\end{align}
%where, $\tilde{I}_{\mu}(\mu)$ is defined in \eqref{eq: tilde I rate function}. 
\end{corollary}
We now specify the operator $\tilde{\Phi}: \mathcal{U} \to \mathcal{U}$. Fix $\mu \in \mathcal{U}$ and defined $\tilde{\Phi}(\mu)$ to be the law of processes $\big(z^p_{[0,T]} , G^p_{[0,T]}\big)_{p\in I_M \fatsemi t\in [0,T]}$. One first defines $\big( G^{p}_t \big)_{p\in I_M \fatsemi t\in [0,T]}$ to be centered Gaussian system such that
\[
\mathbb{E}\big[ G^p_t G^q_s \big] =  \mathbb{E}^{\mu}\big[ \lambda(z^p_t) \lambda(z^q_s) \big].
\]
$(z^p_0)_{p\in I_M}$ is independent of $\big( G^{p}_t \big)_{p\in I_M \fatsemi t\in [0,T]}$ and distributed according to $\hat{\kappa}$. Letting $\big(W^p_{[0,T]}\big)_{p\in I_M}$ be Brownian Motions that are independent of $\mathbf{G}^{\mu}$ , we define $(z^p_t)_{p\in I_M \fatsemi t\in [0,T]}$ to be the strong solution to the stochastic differential equation
\begin{align} \label{eq: SDE conditionedo nG}
dz^p_t = \big( - \tau^{-1} z^p_t + G^{\mu,p}_t \big) dt + \sigma_t dW^p_t.
\end{align}

\begin{theorem}\label{Theorem Coupled Large Deviation Principle Unconditioned}
Assume the connectivity-independent initial conditions (Case 2). For any $\epsilon > 0$,
\begin{align}
\lsup{N} N^{-1} \log \mathbb{P}\big( d_W(\hat{\mu}^N(\mathbf{z},\mathbf{G}) , \xi) \geq \epsilon \big) < 0.
\end{align}
Thus with unit probability,
\begin{equation}
\lim_{N\to\infty}\hat{\mu}^N(\mathbf{z},\mathbf{G})  = \xi .
\end{equation}
Furthermore, 
\begin{align}
\xi = \lim_{n\to\infty} \xi^{(n)}, \label{eq: convergence xi n approximations}
\end{align}
where $\xi^{(n+1)} = \tilde{\Phi}( \xi^{(n)})$ and $\xi^{(1)}$ is any measure in $\mathcal{U}$.
\end{theorem}

\section{Proofs}
We have divided the proofs into three main sections. In Section \ref{Section Regularity Compactness}, we prove general regularity properties of the stochastic processes. In Section \ref{Section LDP Uncoupled}, we prove the LDP for the uncoupled system. In Section \ref{Section Coupled System}, we determine the limiting dynamics of the coupled system.
\subsection{Regularity Estimates and Compactness} \label{Section Regularity Compactness}
We first prove Lemma \ref{Lemma beta nu well defined}.
\begin{proof}
We first check that the covariance function is positive definite (when restricted to a finite set of times). Let $\lbrace t_i \rbrace_{1\leq i \leq m} \subset [0,T]$ be a finite set of times. Then evidently for any constants $\lbrace \alpha^p_i \rbrace_{p\in I_M, 1\leq i \leq m}$, it must be that %restriction of $\eta_{\nu}$ to a finite set of times must be 
\begin{align}
\sum_{p,q\in I_M}\sum_{1\leq i,j \leq m}\alpha^p_i \alpha^q_j \mathbb{E}^{\nu}\big[ \lambda\big(x^p_{t_i}\big) \lambda\big(x^q_{t_j}\big) \big] =  \mathbb{E}^{\nu}\big[ \big( \sum_{p\in I_M}\sum_{1\leq i \leq m}\alpha^p_i \lambda\big(x^p_{t_i}\big) \big)^2 \big]  \geq 0.
\end{align}
This means that there exists a finite set of centered Gaussian variables $\lbrace G^{\nu,p}_{t^{(m)}_i} \rbrace_{p\in I_M \fatsemi 1\leq i \leq m}$ such that \eqref{eq: beta nu covariance definition} holds. It then follows from the Komolgorov Extension Theorem that $\beta_{\nu}$ is well-defined on any countably dense subset of times of $[0,T]$. It remains for us to demonstrate continuity, i.e. that there exists a Gaussian probability law such that \eqref{eq: beta nu covariance definition} holds for all time. We do this using standard theory for the continuity of Gaussian Processes (following Chapter 2 of \cite{Adler2007}).

First, we notice that
\begin{align}
\sup_{p\in I_M}\sup_{t\in [0,T]} \mathbb{E}\big[ (G^{\nu,p}_t)^2 \big] < \infty.
\end{align}
Now define the canonical metric,
\begin{align}
\bar{d}_p(s,t) = &  \mathbb{E}\big[ \big( G^{\nu,p}_s - G^{\nu,p}_t \big)^2\big]^{\frac{1}{2}} =  \mathbb{E}^{\nu}\big[ \big( \lambda(x^{p}_s) - \lambda(x^{p}_t) \big)^2 \big]^{\frac{1}{2}}   \\
\leq &\rm{Const} \sup_{p\in I_M} \mathbb{E}^{\nu}\big[ \big| x^p_s - x^p_t \big|^2 \big]^{\frac{1}{2}} \leq a  \; (t-s)^{\frac{1}{4}} 
\end{align}
thanks to properties of the set $\mathcal{Q}_a$, for all $s,t$ such that $|s-t|$ is smaller than some constant depending on $a$. It follows from Theorem 1.4.1 of \cite{Adler2007} that the Gaussian Process is almost-surely continuous. 

Write $B_t(\epsilon) = \big\lbrace s\in [0,T] : \bar{d}(s,t) \leq \epsilon \big\rbrace$ to be the $\epsilon$-ball about $t$, and let $\mathcal{N}(\epsilon)$ denote the smallest number of such balls that cover $T$. We see that there exists a constant $\mathfrak{c}_a > 0$ such that
\begin{align}
\mathcal{N}(\epsilon) \leq \mathfrak{c}_a\epsilon^{-4}.
\end{align}
Writing $H(\epsilon) = \log \mathcal{N}(\epsilon)$, it follows from Theorem 1.3.5 in \cite{Adler2007} that there exist $M$ Gaussian Processes $(G^{\nu,p}_t)_{t\in [0,T]}$ such that $t \to G^{\nu,p}_t$ is almost-surely continuous, and there exists a universal constant $\mathfrak{K} > 0$ and a random $\eta > 0$ such that for all $\delta < \eta$,
\begin{align}
\sup_{p\in I_M \fatsemi s,t \leq T \fatsemi \bar{d}(s,t) \leq \delta}\big| G^{\nu,p}_s - G^{\nu,p}_t \big| &\leq \mathfrak{K}\int_0^{\delta} H^{1/2}(\epsilon) d\epsilon \\
&\leq \mathfrak{K}\int_0^{\delta} \big(  4\log\big(\epsilon^{-1}\big) + \log \mathfrak{c}_a \big)^{\frac{1}{2}} d\epsilon,
\end{align} 
and we note that the above goes to 0 as $\delta \to 0^+$. This also implies \eqref{eq: uniform continuity U a G}.
\end{proof}
We next prove Corollary \ref{Lemma CompactNess Fields}.
\begin{proof}
The proof of \eqref{eq: uniform continuity U a G conditioned} is analogous to the proof of \eqref{eq: uniform continuity U a G}. Notice that $\mathfrak{m}^p_t(\mu,\mathbf{g})$ depends continuously on $\mathbf{g}$.

 \eqref{eq: uniform continuity U a G} and \eqref{eq: uniform continuity U a G conditioned} imply (respectively)  \eqref{eq: corolalry uniform 1} and \eqref{eq: corolalry uniform 2} .
 \end{proof}

We can now prove Lemma \ref{Lemma Exponential Tightness mu G}.
\begin{proof}
We prove \eqref{Lemma Exp Tightness G} only. The other proof is very similar.

It follows from Lemma \ref{Lemma beta nu well defined} that for any $\epsilon > 0$, there exists a compact set $\mathcal{C}_{\epsilon} \subset \mathcal{C}([0,T],\mathbb{R}^M)$ such that for any $\mu \in \mathcal{Q}_L$, and all $\mathbf{g}_0 \in \mathbb{R}^M$ such that $\norm{\mathbf{g}_0} \leq \epsilon^{-1}$,
\begin{align}
\beta_{\mu,\mathbf{g}_0}\big( G^{\mu}_{[0,T]} \notin \mathcal{C}_{\epsilon} \big) \leq \epsilon. \label{eq: beta nu compact estimate}
\end{align}
It has already been noted above that for any $\lbrace \mathbf{y}_{[0,T]}^j \rbrace_{j\in I_N} \subset   \mathcal{C}([0,T],\mathbb{R}^M)$, $\lbrace \tilde{G}^{j}_{[0,T]} \rbrace_{j\in I_N}$ are independent, and the probability law of $ \tilde{G}^{j}_{[0,T]}$ is $\beta_{\hat{\mu}^N(\mathbf{y}), \mathbf{g}_0 }$. Thus as long as $\hat{\mu}^N(\mathbf{y}) \in \mathcal{Q}_L$, the estimate in \eqref{eq: beta nu compact estimate} holds for any $\tilde{G}^{j}_{[0,T]}$ such that $\| \mathbf{g}^j_0 \| \leq \epsilon^{-1}$.

Define
\begin{align}
\theta_{N,\epsilon}(\mathbf{g}_0) = N^{-1}\sum_{j\in I_N} \chi\lbrace  \| \mathbf{g}^j_0 \| > \epsilon^{-1}  \rbrace .
\end{align}
Our construction of $\mathcal{Y}^N$ implies that  
\begin{align}
\lim_{\epsilon \to 0^+}\lim_{N\to\infty} \sup_{(\mathbf{z}_0,\mathbf{g}_0) \in \mathcal{Y}^N}\theta_{N,\epsilon}(\mathbf{g}_0)  = 0.
\end{align}
Write $I_{N,\epsilon} = \big\lbrace j \; :   \| \mathbf{g}^j_0 \| \leq \epsilon^{-1} \big\rbrace$. Next, for some $\delta > \epsilon/2$, and $b > 0$, by Chernoff's Inequality, for any $N \geq 1$ and any $(\mathbf{y}_0,\mathbf{g}_0) \in \mathcal{Y}^N$,
\begin{align*}
\tilde{\gamma}^N_{\mathbf{y},\mathbf{g}_0} \bigg( N^{-1}\sum_{j\in I_N} \lbrace G^{j}_{[0,T]} \notin \mathcal{C}_{\epsilon} \rbrace \geq \delta \bigg)   &\leq \tilde{\gamma}^N_{\mathbf{y},\mathbf{g}_0} \bigg( N^{-1}\sum_{j\in \tilde{I}_{N,\epsilon}} \lbrace G^{j}_{[0,T]} \notin \mathcal{C}_{\epsilon} \rbrace \geq \delta /2\bigg) \\ &\leq \mathbb{E}^{\tilde{\gamma}^N_{\mathbf{y},\mathbf{g}_0}}\bigg[ \exp\bigg(b\sum_{j\in \tilde{I}_{N,\epsilon}} \lbrace G^{j}_{[0,T]} \notin \mathcal{C}_{\epsilon} \rbrace - bN\delta / 2 \bigg) \bigg] \\
&\leq \big\lbrace \epsilon \exp(b) + 1-\epsilon \big\rbrace^N \exp\big(- N b \delta /2\big)
\end{align*}
We thus find that (by taking small enough $\epsilon$, and $1 \ll b \ll  -\log \epsilon$), for any integer $n$, that there must exist a compact set $\mathcal{C}^{(n)}$ such that for all $N \geq 1$, all $(\mathbf{y}_0,\mathbf{g}_0) \in \mathcal{Y}^N$ and all $\mathbf{y}$ such that $\hat{\mu}^N(\mathbf{y}) \in \mathcal{Q}_L$
\begin{align}
N^{-1}\log  \tilde{\gamma}^N_{\mathbf{y},\mathbf{g}_0} \big( N^{-1}\sum_{j\in I_N} \lbrace G^{j}_{[0,T]} \notin \mathcal{C}^{(n)} \rbrace \geq n^{-1} \big) \leq -n^2.
\end{align}
This motivates us to define the compact set $\tilde{\mathcal{C}}_L \subset \mathcal{P}\big( \mathcal{C}([0,T],\mathbb{R}^M)^2 \big)$ to consist of all measures $\mu$ such that for all $n \geq m$,
\begin{align}
\mu\big( \mathcal{C}^{(n)}\big) \geq 1 - n^{-1}.
\end{align}
Thus using a union-of-events bound,
\begin{align}
 \tilde{\gamma}^N_{\mathbf{y},\mathbf{g}_0}\big( \hat{\mu}^N(\mathbf{G}) \notin\tilde{\mathcal{C}}_L  \big)  &\leq \sum_{n\geq m} \tilde{\gamma}^N_{\mathbf{y},\mathbf{g}_0} \bigg( N^{-1}\sum_{j\in I_N} \lbrace G^{j}_{[0,T]} \notin \mathcal{K}^{(n)} \rbrace \geq n^{-1} \bigg) \\ &\leq \sum_{n\geq m} \exp\big( -Nn^2\big) \\
 &\leq \exp(-NL),
\end{align}
for all $N\geq 1$, as long as $m$ is large enough.

\end{proof}
The following bound on the operator norm of the connectivity matrix is well-known (and the proof is omitted).
\begin{lemma}\label{Lemma bound on operator norm}
For any $L > 0$, there exists $\ell$ such that
\begin{align}
\lsup{N} N^{-1}\log \mathbb{P}\big( \norm{\mathcal{J}_N} \geq \ell \big) \leq -L,
\end{align}
where $\mathcal{J}_N \in \mathbb{R}^{N \times N}$ has $(j,k)$ entry
\[
\mathcal{J}_{N,jk} = N^{-1/2}J^{jk}
\]
\end{lemma}
\begin{lemma} \label{Lemma bound norm of z p j t}
For any $\ell > 0$, there exists $ L > 0$ such that for all $p\in I_M$ and all $N \geq 1$,
\begin{align}
N^{-1}\log \mathbb{P}\big( \mathcal{A}_c \; , \; \sup_{t\in [0,T]}  \sum_{j\in I_N} (z^{p,j}_t)^2 \geq N\ell \big) \leq - L
\end{align}
where
\[
 \mathcal{A}_c = \bigg\lbrace \norm{\mathcal{J}_N} \leq c \; , \;  \sup_{p\in I_M}\sum_{j\in I_N}(z^{p,j}_0)^2 \leq N \mathbb{E}^{\kappa}[(z^p_0)^2] + N   \bigg\rbrace .
 \]

\end{lemma}
\begin{proof}
 Write
 \[
 u_t = N^{-1}\sum_{j\in I_N} (z^{p,j}_t)^2.
 \]
 If the event $\mathcal{A}_c$ holds, then thanks to Ito's Lemma it must be that
 \begin{align}
 du_t =& \big\lbrace- 2\tau^{-1} u_t + 1 +N^{-1} \sum_{j\in I_N} z^{p,j}_t G^{p,j}_t \big\rbrace dt + N^{-1} \sum_{j\in I_N}z^{p,j}_t dW^{p,j}_t \\
 \leq & \big\lbrace- 2\tau^{-1} u_t + 1 +c C_{\lambda} u_t  \big\rbrace dt + N^{-1} \sum_{j\in I_N}z^{p,j}_t dW^{p,j}_t ,
 \end{align}
 since $N^{-1}\sum_{j\in I_N} \lambda(z^{p,j}_t)^2 \leq C_{\lambda}^2 u_t$.  Write
 \begin{align}
 v_t = \sup_{s\in [0,t]} N^{-1} \bigg| \sum_{j\in I_N} \int_0^t z^{p,j}_s dW^{p,j}_s \bigg|,
 \end{align}
 and define the stopping time, for a constant $A > 0$,
 \begin{align}
 \tau_A = \inf\big\lbrace t\geq 0 \; : \; v_t \geq  \exp(At) + A\big\rbrace .
 \end{align}
 Gronwall's Inequality implies that for all $t\leq \tau_A$,
 \begin{align*}
 u_t \leq \big( A + u_0 + t \big) \exp\big( \tilde{c} t \big)
 \end{align*}
 where $\tilde{c} = A + cC_{\lambda} - 2\tau^{-1}$. The quadratic variation of $x(t) := N^{-1} \sum_{j\in I_N} \int_0^t z^{p,j}_s dW^{p,j}_s$ is
 \begin{align}
 (QV)_t^N = N^{-2}\sum_{j\in I_N} \int_0^t (z^{p,j}_s)^2 ds .
 \end{align}
 For all $t\leq \tau_A$,
 \begin{align}
   (QV)_t^N &\leq N^{-1}\tilde{c}^{-1} \big( A + u_0 + t \big) \exp\big( \tilde{c} t \big) := N^{-1}h_t,
 \end{align}
  and notice that $h_t$ is independent of the Brownian Motions. Now define the stochastic process $w(t)$ to be such that
 \begin{align}
w(t) =& x\big( \alpha^N_t \big) \text{ where }\\
\alpha^N_t =& \inf\big\lbrace s \geq 0 \; : \;  (QV)_s^N = t \big\rbrace
\end{align} 
Thanks to the time-rescaled representation of a stochastic integral, $w(t)$ is a Brownian Motion \cite{Karatzas1991}. Writing $f(t) = \exp(At) + A$, it follows that
 \begin{align*}
 \mathbb{P}\bigg( \text{ There exists }&s\leq T \text{ such that }\big|x(s) \big| \geq f(s) \bigg) \\
\leq &\mathbb{P}\bigg( \text{ There exists }s \leq T \text{ such that }\big|w( N^{-1} h_s) \big| \geq f(s) \bigg)\\
\leq &  \mathbb{P}\bigg( \text{ There exists }s \leq T \text{ such that }\big|w( N^{-1} h_{s^{(m)}}) \big| \geq f(s_{(m)}) \bigg)
 \end{align*}
and we have written
\begin{align}
s^{(m)} &= \inf\big\lbrace t^{(m)}_a \; : t^{(m)}_a \geq s \big\rbrace \\
s_{(m)} &= \sup\big\lbrace t^{(m)}_a \; : t^{(m)}_a \leq s \big\rbrace .
\end{align} 
and we recall that $t^{(m)}_a = Ta/m$. Employing a union-of-events bound,
\begin{multline}
 \mathbb{P}\bigg( \text{ There exists }s \leq T \text{ such that }\big|w( h_{s^{(m)}}) \big| \geq f(s_{(m)}) \bigg) \leq \\
 \sum_{a=0}^{m-1} \bigg\lbrace \mathbb{P}\bigg( w\big(N^{-1} h_{t_{a+1}^{(m)}} \big)  \geq f\big(t_{a}^{(m)} \big)  \bigg) 
 + \mathbb{P}\bigg( w\big( N^{-1} h_{t_{a+1}^{(m)}} \big)  \leq - f\big(t_{a}^{(m)} \big)  \bigg) \bigg\rbrace
 \end{multline}
 Now since $w(t)$ is centered and Gaussian, with variance of $t$,
 \begin{align}
N^{-1}\log  \mathbb{P}\bigg( w\big( N^{-1} h_{t_{a+1}^{(m)}} \big)  \geq f\big(t_{a}^{(m)} \big)  \bigg) 
  =&  - \frac{N}{2} f\big(t_{a}^{(m)} \big)^2 \big(  h_{t_{a+1}^{(m)}} \big)^{-1} + O\big( \log N \big) \\
  N^{-1}\log  \mathbb{P}\bigg( w\big( N^{-1} h_{t_{a+1}^{(m)}} \big)  \leq - f\big(t_{a}^{(m)} \big)  \bigg) 
  =&  - \frac{N}{2} f\big(t_{a}^{(m)} \big)^2 \big(  h_{t_{a+1}^{(m)}} \big)^{-1} + O\big( \log N \big).
  \end{align}
 We fix $m=A$, and take $A$ to be arbitrarily large. Then
 \[
\lim_{A\to\infty} \inf_{0\leq a \leq m-1} f\big(t_{a}^{(m)} \big)^2 \big(  h_{t_{a+1}^{(m)}} \big)^{-1} = \infty.
 \]
 We thus find that, for large enough $A$,
 \begin{align}
\lsup{N}N^{-1} \log \mathbb{P}\bigg( \mathcal{A}_c ,\text{ There exists }s\leq T \text{ such that }\big|x(s) \big| \geq f(s) \bigg) \leq - L.
 \end{align}
 We have already demonstrated in the course of the proof that if the event $\mathcal{A}_c$ holds, and $\sup_{s\in [0,T]} |x(s)| \leq f(s)$, then there exists a constant such that $\sup_{t\in[0,T]}u_t \leq \rm{Const}$. We have thus established the Lemma. 
\end{proof}
The following $L^2$-Wasserstein distance provides a very useful way of controlling the dependence of the fields $(G^{\nu}_t)$ on the measure $\nu$. Define $d^{(2)}_t(\cdot,\cdot)$ to be such that for any $\mu,\nu \in \mathcal{U}$,
\begin{align} \label{eq: zeta Wasserstein L squares def}
d^{(2)}_t(\mu,\nu) = \inf_{\zeta} \mathbb{E}^{\zeta}\bigg[ \sum_{p\in I_M}\int_0^t \big\lbrace (y^p_s - \tilde{y}^p_s)^2 + (G^p_s - \tilde{G}^p_s)^2 \big\rbrace ds \bigg]^{1/2},
\end{align}
where the infimum is over all $\zeta \in \mathcal{P}\big( \mathcal{C}([0,T],\mathbb{R}^{2M}) \times  \mathcal{C}([0,T],\mathbb{R}^{2M}\big)$, such that the law of the first $2M$ processes is given by $\mu$, and the law of the last $2M$ processes is given by $\nu$. Let $d^{(2)}(\mu,\nu) := d^{(2)}_T(\mu,\nu) $. %Note that $d^{(2)}_t$ is a pseudo-metric for $ t< T$, and a true metric for $t = T$.
\begin{lemma}\label{Lemma Equivalent Metric}
For any $a > 0$, $d^{(2)}(\cdot,\cdot)$ metrizes weak convergence in $\mathcal{U}_a$. Furthermore,
\begin{align}
\lim_{\epsilon \to 0^+} \sup \big\lbrace d_W(\mu,\nu) \; : \mu,\nu \in \mathcal{U}_a \text{ and }d^{(2)}(\mu,\nu) \leq \epsilon \big\rbrace = 0.
\end{align}
\end{lemma}

\begin{proof}
%\begin{align}
Since $\mathcal{U}_a$ is compact, Prokhorov's Theorem implies that for any $\tilde{\epsilon} > 0$, there exists a compact set $\mathcal{D}_{\epsilon} \subset \mathcal{C}([0,T],\mathbb{R}^M)^2$ such that for all $\mu \in \mathcal{U}_a$,
\begin{align}
\mu\big( \mathcal{D}_{\epsilon} \big) \geq 1 - \tilde{\epsilon}. \label{eq: mu D epsilon bound}
\end{align}
%Suppose that $\mathcal{V} \subset \mathcal{C}([0,T],\mathbb{R}^M)^2$ is compact (with respect to the uniform topology).
Since $\mathcal{D}_{\epsilon}$ is compact, it follows from the Arzela-Ascoli Theorem that for any $\delta > 0$, there exists $\upsilon(\epsilon,\delta)$ such that for all $f,g \in \mathcal{D}_{\epsilon}$ such that for all $p\in I_{2M}$,
\begin{align}
\int_0^T (f^p(t) - g^p(t))^2 dt \leq  \upsilon(\epsilon,\delta),
\end{align}
it necessarily holds that
\begin{align}
\sup_{ p \in I_M }\sup_{t\in [0,T]} \big| f^p(t) - g^p(t) \big| \leq \delta . \label{eq: uniform convergence compact portmanteau}
\end{align}
Let $\zeta$ be any measure that is within $\eta \ll 1$ of realizing the infimum in \eqref{eq: zeta Wasserstein L squares def}. Then, writing
\[
\mathcal{A}_{\epsilon} =  \chi\big\lbrace \text{ For each }p\in I_M, \; \; y^p , \tilde{y}^p , g^p, \tilde{g}^p \in  \mathcal{D}_{\epsilon} \big\rbrace ,
\]
we have the bound
\begin{multline}
\mathbb{E}^{\zeta}\bigg[ \sup_{ p \in I_M }\sup_{t\in [0,T]} \big| y^p(t) - \tilde{y}^p(t) \big|   + \sup_{ p \in I_M }\sup_{t\in [0,T]} \big| g^p(t) - \tilde{g}^p(t) \big| \bigg] \\
\leq \mathbb{E}^{\zeta}\bigg[ \bigg(\sup_{ p \in I_M }\sup_{t\in [0,T]} \big| y^p(t) - \tilde{y}^p(t) \big|   + \sup_{ p \in I_M }\sup_{t\in [0,T]} \big| g^p(t) - \tilde{g}^p(t) \big| \bigg) \mathcal{A}_{\epsilon}  \bigg] +\\
\mathbb{E}^{\zeta}\bigg[ \bigg(\sup_{ p \in I_M }\sup_{t\in [0,T]} \big| y^p(t) - \tilde{y}^p(t) \big|   + \sup_{ p \in I_M }\sup_{t\in [0,T]} \big| g^p(t) - \tilde{g}^p(t) \big| \bigg) \big( 1- \mathcal{A}_{\epsilon} \big)  \bigg] \label{eq: split Portmanteau}
\end{multline}
Now we take $d^{(2)}(\mu,\nu) \to 0^+$, and $\eta \to 0^+$ too. Since $\mathcal{A}_{\epsilon}$ is closed, thanks to the Portmanteau Theorem, we thus find that for any $\epsilon > 0$,
\begin{align}
 \mathbb{E}^{\zeta}\bigg[  \mathcal{A}_{\epsilon} \sum_{p\in I_M}\int_0^T \big\lbrace (y^p_s - \tilde{y}^p_s)^2 + (G^p_s - \tilde{G}^p_s)^2 \big\rbrace ds \bigg] \to 0.
\end{align}
which in turn implies that (making use of the uniform convergence over $\mathcal{D}_{\epsilon}$ in \eqref{eq: uniform convergence compact portmanteau})
\begin{equation}
 \mathbb{E}^{\zeta}\bigg[ \bigg(\sup_{ p \in I_M }\sup_{t\in [0,T]} \big| y^p(t) - \tilde{y}^p(t) \big|   + \sup_{ p \in I_M }\sup_{t\in [0,T]} \big| g^p(t) - \tilde{g}^p(t) \big| \bigg) \mathcal{A}_{\epsilon}  \bigg] \to 0.
\end{equation}
For the other term on the RHS of \eqref{eq: split Portmanteau}, for $b > 0$, write
\begin{align*}
\mathcal{B}_b = \chi \bigg\lbrace  \text{ For each }p\in I_M, \; \sup_{t\in [0,T]}\big| y^p_t \big| \leq b \; , \; \sup_{t\in [0,T]} \big| \tilde{y}^p_t \big| \leq b \; , \; \sup_{t\in [0,T]}\big| g^p_t \big| \leq b \; , \; \sup_{t\in [0,T]} \big| \tilde{g}^p_t \big| \leq b \bigg\rbrace
\end{align*}
Then,
\begin{multline}
\mathbb{E}^{\zeta}\bigg[ \bigg(\sup_{ p \in I_M }\sup_{t\in [0,T]} \big| y^p(t) - \tilde{y}^p(t) \big|   + \sup_{ p \in I_M }\sup_{t\in [0,T]} \big| g^p(t) - \tilde{g}^p(t) \big| \bigg) \big( 1- \mathcal{A}_{\epsilon} \big)  \bigg] \\
\leq \mathbb{E}^{\zeta}\bigg[ \bigg(\sup_{ p \in I_M }\sup_{t\in [0,T]} \big| y^p(t) - \tilde{y}^p(t) \big|   + \sup_{ p \in I_M }\sup_{t\in [0,T]} \big| g^p(t) - \tilde{g}^p(t) \big| \bigg) \big( 1- \mathcal{A}_{\epsilon} \big) \mathcal{B}_b  \bigg] \\
+ \mathbb{E}^{\zeta}\bigg[ \bigg(\sup_{ p \in I_M }\sup_{t\in [0,T]} \big| y^p(t) - \tilde{y}^p(t) \big|   + \sup_{ p \in I_M }\sup_{t\in [0,T]} \big| g^p(t) - \tilde{g}^p(t) \big| \bigg) \big( 1- \mathcal{A}_{\epsilon} \big)\big(1-  \mathcal{B}_b\big)  \bigg] . \label{eq: big d (2) bound}
\end{multline}
Thanks to the fact that, for all $\mu \in \mathcal{U}_a$,
\[
\sup_{p\in I_M} \mathbb{E}^{\mu}\big[ \sup_{t\in [0,T]}(y^p_t)^2 \big] \leq a ,
\]
one finds that the second term on the RHS of \eqref{eq: big d (2) bound} goes to $0$ as $b \to \infty$, uniformly for all $\epsilon > 0$ and all $\mu,\nu \in \mathcal{U}_a$ . For any fixed $b \gg 1$, the first term on the RHS of \eqref{eq: big d (2) bound} must go to zero as $\epsilon \to 0^+$, thanks to \eqref{eq: mu D epsilon bound}. We have thus proved the Lemma.
\end{proof}
For $\mu,\in \mathfrak{Q}$, we define $d^{(2)}_t(\mu,\nu)$ analogously to \eqref{eq: zeta Wasserstein L squares def}.
\begin{lemma} \label{beta mapping continuous}
There exists a constant $\mathfrak{C} > 0$ such that for all $\mu,\nu \in \mathfrak{Q}$ and all $t\in [0,T]$,
\begin{align}
d^{(2)}_t(\beta_\nu , \beta_{\mu}) &\leq \mathfrak{C} d^{(2)}_t(\nu,\mu). \label{eq: d (2) beta nu beta mu}
\end{align}
Also for all $\mu,\nu \in \mathfrak{Q}$ such that for some $b > 0$, $\det( \mathfrak{V}_{\mu,0} ) , \det( \mathfrak{V}_{\nu,0} ) \geq b > 0$, there exists a constant $\mathfrak{C}_{b}$ such that
\begin{align}\label{eq: uniform d (2) conditoined}
d^{(2)}_t\big(\beta_{\nu,\mathbf{g}} , \beta_{\mu, \mathbf{g}}\big) &\leq \tilde{\mathfrak{C}}_{b} (1+ \norm{\mathbf{g}}) d^{(2)}_t(\nu,\mu),
\end{align}
and $\norm{\cdot}$ is the Euclidean norm on $\mathbb{R}^M$.
%(i) The map $\mu \to \beta_{\mu}$ is uniformly continuous over $\mathcal{U}_a$ (for any $a> 0$). \\
%(ii) The map $\mu \to \beta_{\mu,\mathbf{g}_0}$ is uniformly continuous over $\mathcal{U}_a$ (for any $a> 0$). 
\end{lemma}
\begin{proof}
%Thanks to the uniform continuity in Lemma \ref{Lemma beta nu well defined}, if $\lbrace \mu_i \rbrace_{i \geq 1} \subseteq \mathcal{U}_a$ is any sequence such that $d_W(\mu_i,\nu) \to 0$ for some $\nu \in \mathcal{U}_a$, then to show that $\beta_{\mu_i} \to \beta_{\nu}$, it suffices to show that
%\begin{align}
%d^{(2)}(\beta_{\mu_i} , \beta_{\nu}) \to 0.
%\end{align}
We start by proving \eqref{eq: d (2) beta nu beta mu}. Recalling the definition of the distance $d^{(2)}_t$ in \eqref{eq: zeta Wasserstein L squares def}, let $\zeta_{\epsilon}$ be such that
\begin{align}
d^{(2)}_t(\mu,\nu)^2 \geq \mathbb{E}^{\zeta_{\epsilon}}\bigg[ \sum_{p\in I_M} \int_0^t(z^p_s - y^p_s)^2 ds \bigg] + \epsilon .
\end{align}
Furthermore define centered Gaussian processes $G^{(\epsilon),\mu,p}_{s} , G^{(\epsilon),\nu,p}_s$ to be such that for any $p,q\in I_M$ and any $s,t\in [0,T]$,
\begin{align}
\mathbb{E}\big[G^{(\epsilon),\mu,p}_s G^{(\epsilon),\nu,q}_t \big] &= \mathbb{E}^{\zeta_{\epsilon}}\big[ \lambda(z^p_s) \lambda(y^q_t) \big]  \\
\mathbb{E}\big[G^{(\epsilon),\mu,p}_s G^{(\epsilon),\mu,q}_t \big] &= \mathbb{E}^{\zeta_{\epsilon}}\big[ \lambda(z^p_s) \lambda(z^q_t) \big]  \\
\mathbb{E}\big[G^{(\epsilon),\nu,p}_s G^{(\epsilon),\nu,q}_t \big] &= \mathbb{E}^{\zeta_{\epsilon}}\big[ \lambda(y^p_s) \lambda(y^q_t) \big] . 
\end{align}
This definition is possible thanks to a trivial modification of Lemma \ref{Lemma beta nu well defined} (switching $M \to 2M$). We thus find that 
\begin{align}
\lim_{\epsilon \to 0^+}\mathbb{E}\bigg[ \sum_{p\in I_M} \int_0^r (G^{(\epsilon),\mu,p}_t - G^{(\epsilon),\nu,p}_t)^2  dt \bigg] \leq &\lim_{\epsilon \to 0^+} \mathbb{E}^{\zeta_{\epsilon}}\bigg[ \sum_{p\in I_M} \int_0^r\big\lbrace  \lambda(z^p_t) - \lambda(y^p_t) \big\rbrace^2 dt \bigg] \label{eq: Wasserstein G beta nu} \\
\leq & \lim_{\epsilon \to 0^+} C_{\lambda}^2 \mathbb{E}^{\zeta_{\epsilon}}\bigg[ \sum_{p\in I_M} \int_0^r\big\lbrace z^p_t - y^p_t \big\rbrace^2 dt \bigg] \\
=& C_{\lambda}^2 d_r^{(2)}(\mu,\nu).
\end{align}
Now as $\epsilon \to 0^+$,  the LHS of \eqref{eq: Wasserstein G beta nu} must decrease to $d^{(2)}(\beta_\mu, \beta_{\nu})$. \eqref{eq: uniform d (2) conditoined} follows analogously.%We thus find that if $d^{(2)}(\mu,\nu) \to 0$, then necessarily $d^{(2)}(\beta_{\mu_i}, \beta_{\nu}) \to 0$, as required.%Thanks to the definition of the Wasserstein metric, for any $\epsilon > 0$, there exists a measure $\zeta_{\epsilon} \in \mathcal{P}\big( \big(\mathcal{C}([0,T],\mathbb{R})^M \big)^4 \big)$ such that, writing $\zeta_{\epsilon}$ to be the law of $\mathcal{C}([0,T],\mathbb{R})^M$-valued random variables $\mathbf{z},\mathbf{g},\mathbf{y},\mathbf{h}$, and the law of $(\mathbf{z},\mathbf{g})$ is $\nu$, the law of $(\mathbf{y},\mathbf{h})$ is $\mu$, and 
%\begin{align}
%\mathbb{E}^{\zeta_{\epsilon}}\big[ \norm{\mathbf{z}-\mathbf{y}} + \norm{\mathbf{g}-\mathbf{h}} \big] \leq \epsilon.
%\end{align}

The proof of \eqref{eq: uniform d (2) conditoined} is analogous, since the mean and variance functions in \eqref{eq: W definition beta} and \eqref{eq: conditional gaussian 1 beta} are such that for all $\mu,\nu \in \mathcal{U}$, there is a constant $C_b > 0$ such that
\begin{align*}
\sup_{p,q\in I_M \fatsemi s.t \in [0,T]}\big| \mathfrak{W}^{\mu,pq}_{st} - \mathfrak{W}^{\nu,pq}_{st} \big| &\leq C_b\lim_{\epsilon \to 0^+} \sup_{r\in I_M \fatsemi s \in [0,T]} \mathbb{E}^{\zeta_{\epsilon}}\big[ \big( \lambda(z^r_s) - \lambda(y^r_s) \big)^2 \big] \\
\sup_{p,q\in I_M \fatsemi s.t \in [0,T]}\big| \mathfrak{m}^{p}_{s}(\mu,\mathbf{g}) - \mathfrak{m}^{p}_{s}(\nu,\mathbf{g}) \big| &\leq C_b\big(1+ \| \mathbf{g} \| \big) \lim_{\epsilon \to 0^+} \sup_{r\in I_M \fatsemi s \in [0,T]} \mathbb{E}^{\zeta_{\epsilon}}\big[ \big( \lambda(z^r_s) - \lambda(y^r_s) \big)^2 \big]
\end{align*}
We have also employed the fact that $\det( \mathfrak{V}_{\mu,0} )$ is uniformly lower-bounded by a positive constant $b$ (as noted in the statement of the Lemma). 
\end{proof}

\subsection{Large Deviations of the Uncoupled System} \label{Section LDP Uncoupled}

Our first aim is to prove a Large Deviation Principle in the case of fields with a frozen interaction structure (in Lemma \ref{Theorem Uncoupled Large Deviation Principle Conditional Time 0} below). This would ordinarily be a trivial application of Sanov's Theorem. However the proof is slightly complicated by the need for the LDP to be uniform with respect to the variables $(\mathbf{z}_0,\mathbf{g}_0) \in \mathbb{R}^{MN} \times \mathbb{R}^{MN}$ that the probability laws are conditioned on.

For any $\mathbf{g}_0 \in \mathbb{R}^{MN}$ and $\nu \in \mathcal{U}$, define $\tilde{\gamma}^{N}_{\nu , \mathbf{g}_0} := \otimes_{j=1}^N \beta_{\nu, \mathbf{g}_0^j} \in  \mathcal{P}\big( \mathcal{C}([0,T],\mathbb{R}^M)^N \big)$. In other words, $\tilde{\gamma}^{N}_{\nu , \mathbf{g}_0}$ is the law of $N$ independent $\mathcal{C}([0,T],\mathbb{R}^M)$-valued Gaussian variables $\lbrace \tilde{G}^{\nu,j} \rbrace_{j\in I_N}$. The mean and variance of these variables is specified in \eqref{eq: mps average beta} and \eqref{eq: Vps average beta}.  

Let $Q^N_{\nu,\mathbf{z}_0,\mathbf{g}_0} \in \mathcal{P}\big(  \mathcal{C}([0,T],\mathbb{R}^M)^N \times \mathcal{C}([0,T],\mathbb{R}^M)^N \big)$ be the joint law of the uncoupled system, i.e.
\begin{equation}
Q^N_{\nu,\mathbf{z}_0,\mathbf{g}_0} =  \gamma^N_{\nu,\mathbf{z}_0, \mathbf{g}_0} \otimes P^N_{\mathbf{z}_0} .
\end{equation}
For $\nu \in \mathcal{U}$, define $\tilde{I}_{\nu}: \mathcal{P}\big( \mathcal{C}([0,T],\mathbb{R}^M)^2 \big) \to \mathbb{R}$ as follows. We specify that $\tilde{I}_{\nu}(\mu) = \infty$ if either the marginal of $\mu$ at time $0$ is not equal to $\kappa$, and / or $\mu \notin \mathcal{U}$. Otherwise, for any $\zeta \in \mathcal{U}$, writing $\zeta_{\mathbf{z}_0,\mathbf{g}_0}$ to be the law of $\zeta$, conditioned on the values of its variables at time $0$, define
\begin{align}\label{eq: tilde I rate function}
\tilde{I}_{\nu}(\zeta) =  \mathbb{E}^{\kappa}\big[ \mathcal{R}\big(\zeta_{\mathbf{z}_0,\mathbf{g}_0} || P_{\mathbf{z}_0} \otimes \tilde{\gamma}_{\nu,\mathbf{g}_0} \big) \big].
\end{align}
Define the empirical measure $\tilde{\mu}^N \in  \mathcal{P}\big( \mathcal{C}([0,T],\mathbb{R})^M \times  \mathcal{C}([0,T],\mathbb{R})^M \big)$ to be
\begin{align}
\tilde{\mu}^N = N^{-1}\sum_{j\in I_N} \delta_{\mathbf{y}^j_{[0,T]} , \tilde{G}^{\nu,j}_{[0,T]}},
\end{align}
where we recall that
\begin{align}
y^{p}_t = z^{p}_0 +  \int_0^t \sigma_s dW_s^{p}. \label{eq: Y SDE}
\end{align}

\begin{lemma}\label{Theorem Uncoupled Large Deviation Principle Conditional Time 0}
Fix some $\nu \in   \mathcal{U}$. Let $\mathcal{A},\mathcal{O}\subseteq \mathcal{P}\big( \mathcal{C}([0,T],\mathbb{R}^M)^2\big) \big)$, such that $\mathcal{O}$ is open and $\mathcal{A}$ closed. Then
\begin{align}
 \lsup{N}\sup_{(\mathbf{z}_0,\mathbf{g}_0) \in \mathcal{Y}^N} N^{-1}\log Q^N_{\nu,\mathbf{z}_0,\mathbf{g}_0}\big( \tilde{\mu}^N(\mathbf{y}_{[0,T]}, \tilde{\mathbf{G}}^{\nu}_{[0,T]}) \in \mathcal{A}\big) &\leq -\inf_{\mu \in \mathcal{A}}\tilde{I}_{\nu}(\mu) \label{eq: to show LDP nu closed sets}
 \\
\linf{N}\inf_{(\mathbf{z}_0,\mathbf{g}_0) \in \mathcal{Y}^N} N^{-1}\log Q^N_{\nu,\mathbf{z}_0,\mathbf{g}_0} \big( \tilde{\mu}^N(\mathbf{y}_{[0,T]}, \tilde{\mathbf{G}}^{\nu}_{[0,T]}) \in \mathcal{O}\big) &\geq -\inf_{\mu \in \mathcal{O}}\tilde{I}_{\nu}(\mu) . \label{eq: to show LDP nu open sets}
\end{align}
Furthermore $\tilde{I}_{\nu}(\cdot)$ is lower semi-continuous, and has compact level sets.
\end{lemma}
\begin{proof}
First, fix any sequence $\big(\mathbf{z}^{(N)}_0,\mathbf{g}^{(N)}_0\big)_{N\geq 1} $, such that $\big(\mathbf{z}^{(N)}_0,\mathbf{g}^{(N)}_0\big) \in \mathcal{Y}^N$. Necessarily, thanks to the definition of $\mathcal{Y}^N$, it must be that
\begin{align}
\hat{\mu}^N\big(\mathbf{z}^{(N)}_0,\mathbf{g}^{(N)}_0\big) \to \kappa \in \mathcal{P}\big(\mathbb{R}^{2M}\big). \label{eq: emp 0 convergence 0}
\end{align}
It follows from \eqref{eq: emp 0 convergence 0} that
\begin{align}
 \lsup{N} N^{-1}\log Q^N_{\nu,\mathbf{z}^{(N)}_0,\mathbf{g}^{(N)}_0}\big( \tilde{\mu}^N(\mathbf{y}_{[0,T]}, \tilde{\mathbf{G}}^{\nu}_{[0,T]}) \in \mathcal{A}\big) &\leq -\inf_{\mu \in \mathcal{A}}\tilde{I}_{\nu}(\mu) \label{LDP upper tmp} \\
\linf{N}N^{-1}\log  Q^N_{\nu,\mathbf{z}^{(N)}_0,\mathbf{g}^{(N)}_0} \big( \tilde{\mu}^N(\mathbf{y}_{[0,T]}, \tilde{\mathbf{G}}^{\nu}_{[0,T]}) \in \mathcal{O}\big) &\geq -\inf_{\mu \in \mathcal{O}}\tilde{I}_{\nu}(\mu) . \label{eq: any sequence LDP}
\end{align}
See for instance \cite{Lucon2017} for a proof of this fact. Furthermore $\tilde{I}_{\nu}: \mathcal{P}\big( \mathcal{C}([0,T],\mathbb{R}^M)^2\big) \to \mathbb{R}^+$ is lower-semi-continuous and has compact level sets.

We next have to show that the convergence is uniform over $\mathcal{Y}^N$ (as in the statement of the Theorem). To do this, we first wish to show that for any measurable set $\mathcal{E} \subset \mathcal{P}\big( \mathcal{C}([0,T],\mathbb{R}^M)^2 \big)$ and any $\epsilon > 0$, for all large enough $N$,
\begin{multline}
\sup_{(\mathbf{z}_0,\mathbf{g}_0) \in \mathcal{Y}^N} Q^N_{\nu,\mathbf{z}_0,\mathbf{g}_0}\big( \tilde{\mu}^N(\mathbf{y}_{[0,T]}, \tilde{\mathbf{G}}^{\nu}_{[0,T]}) \in \mathcal{E}\big) \leq  \inf_{(\mathbf{z}_0,\mathbf{g}_0) \in \mathcal{Y}^N} Q^N_{\nu,\mathbf{z}_0,\mathbf{g}_0}\big( \tilde{\mu}^N(\mathbf{y}_{[0,T]}, \tilde{\mathbf{G}}^{\nu}_{[0,T]}) \in \mathcal{E}^{(\epsilon)}\big) \label{eq:to show epsilon blowup}
\end{multline}
and $\mathcal{E}^{(\epsilon)}$ is the closed $\epsilon$-blowup of $\mathcal{E}$ with respect to $d_W$. To do this, we are going to compare the conditioned probability to the conditioned probability induced by any other sequence in $\mathcal{Y}^N$. This comparison is facilitated by using the following permutation-averaged probability law.

Define the set
\begin{multline}
\mathfrak{S}^N = \bigg\lbrace (\mathbf{y},\mathbf{g}) \in \mathcal{C}\big( [0,T], \mathbb{R}^M \big)^N \times  \mathcal{C}\big( [0,T], \mathbb{R}^M \big)^N:  
\\ N^{-1} \sum_{j\in I_N}\big\lbrace \sup_{p\in I_M}\big| y^{p,j}_0 - z^{(N),p,j}_0 \big| + \sup_{p\in I_M}\big| g^{p,j}_0 - g^{(N),p,j}_0 \big| \big\rbrace \leq 2 \tilde{\delta}_N \bigg\rbrace,
\end{multline} 
and we recall that $(\tilde{\delta}_N)_{N\geq 1}$ is a sequence that decreases to $0$, as defined in \eqref{eq: Y N set}. We endow $\mathfrak{S}^N$ with the topology that it inherits from $\mathcal{C}\big( [0,T], \mathbb{R}^M \big)^N \times  \mathcal{C}\big( [0,T], \mathbb{R}^M \big)^N$. Write $\mathfrak{P}^N$ to be the set of all permutations on $I_N$, and define the measure $\bar{Q}^N_{\nu,\mathbf{z}_0,\mathbf{g}_0} \in \mathcal{P}\big( \mathfrak{S}^N \big)$ to be the average over all permutations, i.e. for any measurable $\mathcal{A} \subseteq \mathfrak{S}^N$,
\begin{align}
\bar{Q}^N_{\nu,\mathbf{z}_0,\mathbf{g}_0}(\mathcal{A})= \big| \mathfrak{P}^N \big|^{-1} \sum_{\pi \in \mathfrak{P}^N} Q^N_{\nu,\mathbf{z}_0,\mathbf{g}_0}\big( \pi \cdot \mathcal{A}\big)
\end{align}
and here we denote $\pi:\mathcal{C}\big( [0,T], \mathbb{R}^M \big)^N \times  \mathcal{C}\big( [0,T], \mathbb{R}^M \big)^N $ to be the permutation,
\begin{align}
\big(\pi\cdot(\mathbf{y},\mathbf{g}) \big)^j = (\mathbf{y}^{\pi(j)},\mathbf{g}^{\pi(j)}).
\end{align}
Since the empirical measure is invariant under any permutation of its arguments, for any measurable $\mathcal{A} \subset \mathcal{P}\big(\mathcal{C}([0,T],\mathbb{R}^M)^2 \big)$
\begin{align}
Q^N_{\nu,\mathbf{z}^{(N)}_0,\mathbf{g}^{(N)}_0}\big( \tilde{\mu}^N(\mathbf{y}_{[0,T]}, \tilde{\mathbf{G}}^{\nu}_{[0,T]}) \in \mathcal{A}\big) = \bar{Q}^N_{\nu,\mathbf{z}^{(N)}_0,\mathbf{g}^{(N)}_0}\big( \tilde{\mu}^N(\mathbf{y}_{[0,T]}, \tilde{\mathbf{G}}^{\nu}_{[0,T]}) \in \mathcal{A}\big) .
\end{align}

We can without loss of generality take $\big(\mathbf{z}^{(N)}_0,\mathbf{g}^{(N)}_0\big) \in \mathfrak{S}^N$. Now consider any other sequence $\big(\grave{\mathbf{z}}^{(N)}_0,\grave{\mathbf{g}}^{(N)}_0\big) \in \mathfrak{S}^N \cap \mathcal{Y}^N$. Let $d_W$ be the $1$-Wasserstein Distance on $\mathcal{P}\big( \mathcal{C}([0,T],\mathbb{R}^M)^N \times  \mathcal{C}([0,T],\mathbb{R}^M)^N \big)$ induced by the  norm
\begin{align}
\norm{ (\mathbf{y},\mathbf{G}) - (\bar{ \mathbf{y}},\bar{\mathbf{G}}) } = N^{-1}\sum_{j\in I_N}\big\lbrace \sup_{p\in I_M \fatsemi t\in [0,T]} \big| y^{p,j}_t - \bar{y}^{p,j}_t \big| + \sup_{p\in I_M \fatsemi t\in [0,T]} \big| G^{p,j}_t - \bar{G}^{p,j}_t \big| \big\rbrace .
\end{align}
We claim that
\begin{align}
d_W\big( \bar{Q}^N_{\nu,\mathbf{z}^{(N)}_0,\mathbf{g}^{(N)}_0},  \bar{Q}^N_{\nu,\grave{\mathbf{z}}^{(N)}_0,\grave{\mathbf{g}}^{(N)}_0} \big) \leq & N^{-1}\sum_{j\in I_N} \big\lbrace \sup_{p\in I_M} \big| z^{(N),p,j}_0 - \grave{z}^{(N),p,j}_0 \big| \nonumber\\ &+  \sup_{p\in I_M \fatsemi t\in [0,T]} \big|    \mathfrak{m}^{p}_{t}(\nu,\mathbf{g}^{(N),j}_0) - \mathfrak{m}^{p}_{t}(\nu,\grave{\mathbf{g}}^{(N),j}_0)  \big| \big\rbrace  \\
:= & f^N\big( \mathbf{z}^{(N)}_0,\mathbf{g}^{(N)}_0 , \grave{\mathbf{z}}^{(N)}_0,\grave{\mathbf{g}}^{(N)}_0 \big)
\end{align}
This identity follows from the fact that $\tilde{z}^{p,j}_t := z^{p,j}_t - z^{p,j}_0$ and $\tilde{g}^{p,j}_t := G^{p,j}_t - \mathfrak{m}^p_t(\nu, \mathbf{G}^j_0) $ are identically distributed, both (i) for all $j\in I_N$, and (ii) with respect to both probability laws $ \bar{Q}^N_{\nu,\mathbf{z}^{(N)}_0,\mathbf{g}^{(N)}_0}$ and $ \bar{Q}^N_{\nu,\grave{\mathbf{z}}^{(N)}_0,\grave{\mathbf{g}}^{(N)}_0} $. 

It follows from the definition of $\mathcal{Y}^N$ that  
\begin{align}
\lim_{N\to\infty}\sup_{(\mathbf{z}^{(N)}_0,\mathbf{g}^{(N)}_0), (\grave{\mathbf{z}}^{(N)}_0,\grave{\mathbf{g}}^{(N)}_0) \in \mathfrak{S}^N \cap \mathcal{Y}^N}  f^N\big( \mathbf{z}^{(N)}_0,\mathbf{g}^{(N)}_0 , \grave{\mathbf{z}}^{(N)}_0,\grave{\mathbf{g}}^{(N)}_0 \big) = 0.
\end{align}
We have thus proved \eqref{eq:to show epsilon blowup}.

To now must prove the Large Deviations bounds in the statement of the theorem. We start with the upper-bound \eqref{eq: to show LDP nu closed sets}. It follows from \eqref{LDP upper tmp} and \eqref{eq:to show epsilon blowup} that for any $\epsilon > 0$,
\begin{align}
 \lsup{N}\sup_{(\mathbf{z}_0,\mathbf{g}_0) \in \mathcal{Y}^N} N^{-1}\log Q^N_{\nu,\mathbf{z}_0,\mathbf{g}_0}\big( \tilde{\mu}^N(\mathbf{y}_{[0,T]}, \tilde{\mathbf{G}}^{\nu}_{[0,T]}) \in \mathcal{A}\big) \leq - \inf_{\mu \in \mathcal{A}^{(\epsilon)}} \tilde{I}_{\nu}(\mu).
 \end{align}
 The lower-semi-continuity of $\tilde{I}_{\nu}(\mu)$ dictates that
 \[
\lim_{\epsilon \to 0^+} \inf_{\mu \in \mathcal{A}^{(\epsilon)}} \tilde{I}_{\nu}(\mu) = \inf_{\mu \in \mathcal{A}} \tilde{I}_{\nu}(\mu),
 \]
 and we have proved the upperbound. For the lower-bound, let $\mathcal{O}$ be open, and for any $\mu \in \mathcal{O}$, take $\epsilon > 0$ to be such that $B_{2\epsilon}(\mu) \subset \mathcal{O}$. Then 
 \begin{align}
 \linf{N}\inf_{(\mathbf{z}_0,\mathbf{g}_0) \in \mathcal{Y}^N} &N^{-1}\log Q^N_{\nu,\mathbf{z}_0,\mathbf{g}_0} \big( \tilde{\mu}^N(\mathbf{y}_{[0,T]}, \tilde{\mathbf{G}}^{\nu}_{[0,T]}) \in \mathcal{O}\big)  \\
&\geq  \linf{N}\inf_{(\mathbf{z}_0,\mathbf{g}_0) \in \mathcal{Y}^N} N^{-1}\log Q^N_{\nu,\mathbf{z}_0,\mathbf{g}_0} \big( \tilde{\mu}^N(\mathbf{y}_{[0,T]}, \tilde{\mathbf{G}}^{\nu}_{[0,T]}) \in B_{2\epsilon}(\mu) \big) \\
&\geq  \linf{N} N^{-1}\log Q^N_{\nu,\mathbf{z}^{(N)}_0,\mathbf{g}^{(N)}_0} \big( \tilde{\mu}^N(\mathbf{y}_{[0,T]}, \tilde{\mathbf{G}}^{\nu}_{[0,T]}) \in B_{\epsilon}(\mu) \big) \\
 &\geq -\inf_{\zeta \in B_{\epsilon}(\mu)}\tilde{I}_{\nu}(\zeta),
\end{align}
using the Large Deviations estimate \eqref{eq: any sequence LDP}. Taking $\epsilon \to 0^+$, it must be that for any $\mu \in \mathcal{O}$,
 \begin{align}
 \linf{N}\inf_{(\mathbf{z}_0,\mathbf{g}_0) \in \mathcal{Y}^N} N^{-1}\log Q^N_{\nu,\mathbf{z}_0,\mathbf{g}_0} \big( \tilde{\mu}^N(\mathbf{y}_{[0,T]}, \tilde{\mathbf{G}}^{\nu}_{[0,T]}) \in \mathcal{O}\big) \geq - \tilde{I}_{\nu}(\mu).
\end{align}
Since $\mu \in \mathcal{O}$ is arbitrary, \eqref{eq: to show LDP nu open sets} follows immediately.
\end{proof}

We can now state the proof of Theorem \ref{Theorem Uncoupled Large Deviation Principle Conditional Time 0 0}.
\begin{proof}
We start with the upper bound \eqref{LDP upper bound}. We write $\hat{\mu}^N := \hat{\mu}^N(\mathbf{y}_{[0,T]}, \tilde{\mathbf{G}}_{[0,T]})$. Using a union-of-events bound, for any $a > 0$,
\begin{multline}
\lsup{N}\sup_{(\mathbf{z}_0,\mathbf{g}_0) \in \mathcal{Y}^N} N^{-1}\log Q^N_{\mathbf{z}_0,\mathbf{g}_0}\big(\hat{\mu}^N  \in \mathcal{A}\big)  \leq \\ \max\bigg\lbrace\lsup{N}\sup_{(\mathbf{z}_0,\mathbf{g}_0) \in \mathcal{Y}^N} N^{-1}\log Q^N_{\mathbf{z}_0,\mathbf{g}_0}\big( \hat{\mu}^N  \in \mathcal{A} \cap \mathcal{U}_a \big), \\ \lsup{N}\sup_{(\mathbf{z}_0,\mathbf{g}_0) \in \mathcal{Y}^N} N^{-1}\log Q^N_{\mathbf{z}_0,\mathbf{g}_0}\big( \hat{\mu}^N  \notin \mathcal{U}_a \big) \bigg\rbrace  \\
\leq  \max\bigg\lbrace\lsup{N}\sup_{(\mathbf{z}_0,\mathbf{g}_0) \in \mathcal{Y}^N} N^{-1}\log Q^N_{\mathbf{z}_0,\mathbf{g}_0}\big( \hat{\mu}^N  \in \mathcal{A} \cap \mathcal{U}_a \big), -L \bigg\rbrace,
\end{multline}
for any $L > 0$, as long as $a$ is sufficiently large, thanks to the exponential tightness proved in Lemma \ref{Lemma Exponential Tightness}. By taking $a \to \infty$, it thus suffices that we prove that for arbitrary $\mathcal{U}_a$ such that $\mathcal{A} \cap \mathcal{U}_a \neq \emptyset$,
\begin{align}
\lsup{N}\sup_{(\mathbf{z}_0,\mathbf{g}_0) \in \mathcal{Y}^N} N^{-1}\log Q^N_{\mathbf{z}_0,\mathbf{g}_0}\big( \hat{\mu}^N  \in \mathcal{A} \cap \mathcal{U}_a \big) = -\inf_{\mu \in \mathcal{A} \cap \mathcal{U}_a} \tilde{I}_{\mu}(\mu). \label{eq: to show intersected closed set}
\end{align}

%For any $\epsilon > 0$ and $\mu \in  \mathcal{P}\big( \mathcal{C}([0,T],\mathbb{R}^M)\big) $, define
%\begin{align}
%\breve{B}_{\epsilon}(\mu) = \bigg\lbrace \nu \in   \mathcal{P}\big( \mathcal{C}([0,T],\mathbb{R}^M)^2\big) \; : \; (i) \nu^{(1)} = \mu\rm{ and } (ii)  d_W(\mu^{(,\nu) < \epsilon \bigg\rbrace
%\end{align}
%Using Lemma \ref{Lemma Unifrom Approximation}
%For any $\zeta \in  \mathcal{P}\big( \mathcal{C}([0,T],\mathbb{R}^M)^2\big)$, define
%\[
%\hat{B}_{\epsilon}(\zeta) = \bigg\lbrace \mu \in  \mathcal{P}\big( \mathcal{C}([0,T],\mathbb{R}^M)^2\big)  \; : \; d_W(\mu^{(1)},\zeta^{(1)}) < \epsilon \text{ and } d_W(\mu^{(1)},\zeta^{(1)}) < \epsilon \bigg\rbrace .
%\]
Since $\mathcal{A} \cap \mathcal{U}_a$ is compact, for any $\epsilon > 0$ we can always find an open covering of the form, for some positive integer $\mathcal{N}_{\epsilon}$,  $\lbrace \zeta_i \rbrace_{1\leq i \leq \mathcal{N}_{\epsilon}} \subseteq \mathcal{A} \cap \mathcal{U}_a$,
\begin{align}
\mathcal{A} \cap \mathcal{U}_a \subseteq \bigcup_{i=1}^{\mathcal{N}_{\epsilon}} B_{\epsilon}(\zeta_i).
\end{align}
We thus find that 
\begin{multline}
\lsup{N}\sup_{(\mathbf{z}_0,\mathbf{g}_0) \in \mathcal{Y}^N} N^{-1}\log Q^N_{\mathbf{z}_0,\mathbf{g}_0}\big( \hat{\mu}^N  \in \mathcal{A} \cap \mathcal{U}_a \big) \\ \leq \sup_{1\leq i \leq \mathcal{N}_{\epsilon}} \bigg\lbrace \lsup{N}\sup_{(\mathbf{z}_0,\mathbf{g}_0) \in \mathcal{Y}^N} N^{-1}\log Q^N_{\mathbf{z}_0,\mathbf{g}_0}\big( \hat{\mu}^N \in B_{\epsilon}(\zeta_i) \big) \bigg\rbrace .
\end{multline}
Thus, employing Lemma \ref{Lemma beta nu well defined} in the third line below,
\begin{align}
\lsup{N}\sup_{(\mathbf{z}_0,\mathbf{g}_0) \in \mathcal{Y}^N}& N^{-1} \log Q^N_{\mathbf{z}_0,\mathbf{g}_0}\big( \hat{\mu}^N \in B_{\epsilon}(\zeta_i) \big) \\ &=\lsup{N}\sup_{(\mathbf{z}_0,\mathbf{g}_0) \in \mathcal{Y}^N} N^{-1} \log  \mathbb{E}^{P^N_{\mathbf{z}_0}} \bigg[ \gamma^N_{\mathbf{y},\mathbf{g}_0}\bigg(  \hat{\mu}^N \in B_{\epsilon}(\zeta_i) \bigg) \bigg] \\
&=\lsup{N}\sup_{(\mathbf{z}_0,\mathbf{g}_0) \in \mathcal{Y}^N} N^{-1} \log  \mathbb{E}^{P^N_{\mathbf{z}_0}} \bigg[ \tilde{\gamma}^N_{\hat{\mu}^N(\mathbf{y}),\mathbf{g}_0}\bigg(  \hat{\mu}^N \in B_{\epsilon}(\zeta_i) \bigg) \bigg] \\
&\leq \lsup{N}\sup_{(\mathbf{z}_0,\mathbf{g}_0) \in \mathcal{Y}^N} N^{-1} \log  \mathbb{E}^{P^N_{\mathbf{z}_0}} \bigg[  \sup_{\nu \in  B_{\epsilon}(\zeta_i)}  \gamma^N_{\nu^{(1)},\mathbf{g}_0}\bigg(  \hat{\mu}^N \in B_{\epsilon}(\zeta_i) \bigg) \bigg] \\
&= - \inf_{\mu , \nu \in B_{\epsilon}(\zeta_i)} \tilde{I}_{\nu}(\mu),
\end{align}
thanks to Theorem \ref{Theorem Uncoupled Large Deviation Principle Conditional Time 0}. We thus find that
\begin{align}
\lsup{N}\sup_{(\mathbf{z}_0,\mathbf{g}_0) \in \mathcal{Y}^N} N^{-1}\log Q^N_{\mathbf{z}_0,\mathbf{g}_0}\big( \hat{\mu}^N  \in \mathcal{A} \cap \mathcal{U}_a \big)  &\leq - \inf_{1\leq i \leq \mathcal{N}_{\epsilon}} \inf_{\nu ,\mu \in B_{\epsilon}(\zeta_i)} \tilde{I}_{\nu}(\mu).\label{eq: l sup infimum expression}% \\&\leq - \tilde{I}_{\nu_{(\epsilon)}}(\mu_{(\epsilon)}) + \epsilon,
\end{align}
Now, it is proved in Lemma \ref{beta mapping continuous} that $ \nu \to Q_{\nu,\mathbf{z}_0,\mathbf{g}_0} $ is continuous. Since the Relative Entropy is lower-semi-continuous in both of its arguments, we thus find that the following map is lower-semi-continuous,
\[
(\nu,\mu) \to \tilde{I}_{\nu}(\mu).
\]
Thus taking $\epsilon \to 0^+$, we obtain that
\begin{align}
\lim_{\epsilon \to 0^+}\inf_{1\leq i \leq \mathcal{N}_{\epsilon}} \inf_{\nu ,\mu \in B_{\epsilon}(\zeta_i)} \tilde{I}_{\nu}(\mu)
= \inf_{\mu \in \mathcal{A} \cap \mathcal{U}_a} \tilde{I}_{\mu}(\mu),
\end{align}
and we have proved \eqref{eq: to show intersected closed set}.

Turning to the lower bound \eqref{LDP lower bound}, consider an arbitrary open set $\mathcal{O}$. If $\mathcal{O} \cap \mathcal{U} = \emptyset$, then 
\[
\lsup{N}\sup_{(\mathbf{z}_0,\mathbf{g}_0) \in \mathcal{Y}^N} N^{-1}\log Q^N_{\mathbf{z}_0,\mathbf{g}_0}\big(\hat{\mu}^N  \in \mathcal{O}\big) = - \infty = - \inf_{\mu \in \mathcal{O}}\mathcal{I}(\mu),
\]
since $\mathcal{I}$ is identically $\infty$ outside of $\mathcal{U}$. In this case, its clear that \eqref{eq: to show LDP nu open sets} holds.

We can thus assume that $\mathcal{O} \cap \mathcal{U} \neq \emptyset$. Let $\mu \in \mathcal{O}$ be such that $\mu$ is in the interior of $\mathcal{U}_a$, for some $a > 0$. We can thus find a sequence of neighborhoods $\lbrace \mathcal{N}_i \rbrace_{i \geq 1}$ of $\mu$ such that $ \mathcal{N}_j \subseteq \mathcal{O}\cap \mathcal{U}_a \cap B_{j^{-1}}(\mu) $. We thus find that for any $j\geq 1$,
\begin{multline}
\lsup{N}\sup_{(\mathbf{z}_0,\mathbf{g}_0) \in \mathcal{Y}^N} N^{-1}\log Q^N_{\mathbf{z}_0,\mathbf{g}_0}\big(\hat{\mu}^N  \in \mathcal{O}\big) \geq\lsup{N}\sup_{(\mathbf{z}_0,\mathbf{g}_0) \in \mathcal{Y}^N} N^{-1}\log Q^N_{\mathbf{z}_0,\mathbf{g}_0}\big(\hat{\mu}^N  \in  \mathcal{N}_j\big) .
\end{multline}
Similarly to the bound for the closed sets, we obtain that 
\begin{align}
\lsup{N}\sup_{(\mathbf{z}_0,\mathbf{g}_0) \in \mathcal{Y}^N} N^{-1}\log Q^N_{\mathbf{z}_0,\mathbf{g}_0}\big(\hat{\mu}^N  \in  \mathcal{N}_j\big) \geq - \sup_{\nu \in \mathcal{N}_j}\inf_{\mu \in \mathcal{N}_j}\tilde{I}_{\nu}(\mu).
\end{align}
Taking $j \to \infty$, since $(\nu,\mu) \to \tilde{I}_{\nu}(\mu) $ is lower semicontinuous, it must be that 
\begin{align}
\lsup{N}\sup_{(\mathbf{z}_0,\mathbf{g}_0) \in \mathcal{Y}^N} N^{-1}\log Q^N_{\mathbf{z}_0,\mathbf{g}_0}\big(\hat{\mu}^N  \in  \mathcal{N}_j\big) \geq - \tilde{I}_{\mu}(\mu).
\end{align}
Since $\mu \in \mathcal{O}$ is arbitrary, it must be that 
\begin{align}
\lsup{N}\sup_{(\mathbf{z}_0,\mathbf{g}_0) \in \mathcal{Y}^N} N^{-1}\log Q^N_{\mathbf{z}_0,\mathbf{g}_0}\big(\hat{\mu}^N  \in  \mathcal{O} \big) \geq - \inf_{\mu \in \mathcal{O}}\tilde{I}_{\mu}(\mu).
\end{align}

\end{proof}
\subsubsection{Uncoupled System (with no conditioning)}
In this subsection we prove Corollary \ref{Corollary Uncoupled Large Deviation Principle Unconditioned}. 

For some $\nu \in \mathfrak{Q}$, let $Q^N_{\nu } \in \mathcal{P}\big(  \mathcal{C}([0,T],\mathbb{R}^M)^N \times \mathcal{C}([0,T],\mathbb{R}^M)^N \big)$ be the joint law of the uncoupled system (with no conditioning), i.e.
\begin{equation}
Q^N_{\nu } =  \big(\beta_{\nu} \otimes P_{\mathbf{z}} \big)^{\otimes N} .
\end{equation}

We reach a corollary to Lemma \ref{Theorem Uncoupled Large Deviation Principle Conditional Time 0}. %The proof is very similar to that of Lemma \ref{Theorem Uncoupled Large Deviation Principle Conditional Time 0}, but also simpler because there are no conditioned variables. The proof is there
\begin{corollary}\label{Theorem Uncoupled Large Deviation Principle Conditional Time 0 Corollary}
Fix some $\nu \in   \mathcal{U}$. Let $\mathcal{A},\mathcal{O}\subseteq \mathcal{P}\big( \mathcal{C}([0,T],\mathbb{R}^M)^2\big) \big)$, such that $\mathcal{O}$ is open and $\mathcal{A}$ closed. Then
\begin{align}
 \lsup{N} N^{-1}\log Q^N_{\nu}\big( \tilde{\mu}^N(\mathbf{y}_{[0,T]}, \tilde{\mathbf{G}}^{\nu}_{[0,T]}) \in \mathcal{A}\big) &\leq -\inf_{\mu \in \mathcal{A}} \mathcal{R}(\mu || S_{\nu} ) \label{eq: to show LDP nu closed sets 2}
 \\
\linf{N} N^{-1}\log Q^N_{\nu } \big( \tilde{\mu}^N(\mathbf{y}_{[0,T]}, \tilde{\mathbf{G}}^{\nu}_{[0,T]}) \in \mathcal{O}\big) &\geq -\inf_{\mu \in \mathcal{O}} \mathcal{R}(\mu || S_{\nu} ) . \label{eq: to show LDP nu open sets 2}
\end{align}
Furthermore $\mu \to \mathcal{R}(\mu || S_{\nu} )$ is lower semi-continuous, and has compact level sets.
\end{corollary}
\begin{proof}
This is a consequence of Sanov's Theorem.
\end{proof}
The proof of Corollary \ref{Corollary Uncoupled Large Deviation Principle Unconditioned} now follows analogously to the proof of  \ref{Theorem Uncoupled Large Deviation Principle Conditional Time 0 0}.
\subsection{Coupled System} \label{Section Coupled System}

Girsanov's Theorem implies that
\begin{align} \label{eq: RN derivative}
\frac{dP^N_{\mathbf{J},\mathbf{z}_0}}{dP^N_{\mathbf{z}_0}}\bigg|_{\mathcal{F}_T}(\mathbf{y}) = \exp\big( N \Gamma^N_{\mathbf{J},T}(\mathbf{y})\big)
\end{align}
where $ \Gamma^N_{\mathbf{J},T}: \mathbb{R}^{MN} \to \mathbb{R}$ is
\begin{equation}
 \Gamma^N_{\mathbf{J},T}(\mathbf{y})= N^{-1}\sum_{j\in I_N \fatsemi p\in I_M} \int_0^T \sigma_s^{-2}  \big(\tilde{G}^{p,j}_s- \tau^{-1} y^{p,j}_s   \big)dy^{p,j}_s -\frac{1}{2} \sigma_s^{-2} \big(\tilde{G}^{p,j}_s- \tau^{-1}y^{p,j}_s   \big)^2 ds .%\sum_{p\in I_M} \int_0^u \sigma_s^{-2} \mathbb{E}^{\mu}\bigg[  \big(g^{p}_s- \tau^{-1} z^p_s   \big)dz^{p}_s -\frac{1}{2} \big(g^{p}_s- \tau^{-1}z^p_s   \big)^2 \bigg] ds .
\end{equation}
We wish to specify a map $\Gamma:  \mathcal{U} \to \mathbb{R}$ with (i) as nice regularity properties as possible, and (ii) such that with unit probability
\begin{align} \label{eq: Radon-Nikodym N system}
 \Gamma^N_{\mathbf{J},T}(\mathbf{y}) = \Gamma\big( \hat{\mu}^N(\mathbf{y},\tilde{\mathbf{G}} ) \big).
\end{align}
It is well-known that the stochastic integral is not a continuous function of the driving Brownian motion. Thus we define the map $\Gamma$ to be a limit of time-discretized approximations, and we will show that this limit must always converge for any measure in $\mathcal{U}$. 

Our precise definition of $\Gamma: \mathcal{U} \to \mathbb{R}$ is as follows. We first define a time-discretized approximation to $\Gamma$. $\Gamma^{(m)} : \mathcal{U} \to \mathbb{R}^+$,
\begin{multline}
\Gamma^{(m)}(\mu) = \sum_{p\in I_M} \sum_{a=0}^{m-1} \mathbb{E}^{\mu}\bigg[ \sigma_{t^{(m)}_a}^{-2} \big(  G^p_{t_a^{(m)}} - \tau^{-1} z^{p}_{t_a^{(m)}}\big) \big( z^{p}_{t^{(m)}_{a+1}} -  z^{p}_{t^{(m)}_{a}} + \Delta_m \tau^{-1} z^{p}_{t^{(m)}_{a}} \big) \\- \frac{1}{2} \sigma_{t^{(m)}_a}^{-2} \Delta_m \big(  G^p_{t_a^{(m)}} - \tau^{-1} z^{p}_{t_a^{(m)}}\big)^2   \bigg].
\end{multline}
We now define $\Gamma: \mathcal{U} \to \mathbb{R}$ to be such that (in the case that the following limit exists)
\begin{align}
\Gamma(\mu) = \lim_{j\to \infty} \Gamma^{(m_{j,j})}(\mu), \label{Gamma mu limit}
\end{align}
where $m_{j,j}$ is a positive integer defined further below in Lemma \ref{Lemma discrete m n approximation}. If the above limit does not exist, then we define $\Gamma(\mu) = 0$ (in fact we will see that the limit always exists if $\mu \in \mathcal{U}$). It may be observed that $\Gamma$ is a well-defined measurable function.

\begin{lemma}\label{Lemma Approximate Gamma by Continuous Functions}
For every $N\geq 1$, every $(\mathbf{z}_0,\mathbf{g}_0) \in \mathcal{Y}^N$, and for $Q^N_{\mathbf{z}_0,\mathbf{g}_0}$ almost every $(\mathbf{y},\tilde{\mathbf{G}})$, the following limit exists
\begin{equation}
\lim_{j\to \infty} \Gamma^{(m_{j,j})} \big( \hat{\mu}^N(\mathbf{y},\tilde{\mathbf{G}}) \big) \label{eq: empirical measure limit}
\end{equation}
With unit probability, the Radon-Nikodym Derivative in \eqref{eq: RN derivative} is such that
\begin{align} 
\frac{dP^N_{\mathbf{J},\mathbf{z}_0}}{dP^N_{\mathbf{z}_0}}\bigg|_{\mathcal{F}_T}  = \exp\big(  \Gamma\big( \hat{\mu}^N(\mathbf{y},\tilde{\mathbf{G}}) \big) \big)
\end{align}
Also for any $\epsilon,L > 0$, there exists $k\in \mathbb{Z}^+$ such that for all $N \geq 1$,
\begin{align} \label{eq: uniform approximation RD Derivative}
\sup_{j \geq k}\sup_{\mathbf{z}_0,\mathbf{g}_0 \in \mathcal{Y}^N} N^{-1}\log Q^N_{\mathbf{z}_0,\mathbf{g}_0} \bigg( \big|  \Gamma^{(m_{j,j})} \big( \hat{\mu}^N(\mathbf{y},\tilde{\mathbf{G}}) \big)  - \Gamma \big( \hat{\mu}^N(\mathbf{y},\tilde{\mathbf{G}}) \big) \big| \geq \epsilon  \bigg) \leq - L
\end{align}
\end{lemma}
\begin{proof}
Define the set
\begin{align}
\mathcal{A}_j = \big\lbrace \mu \in \mathcal{U}:  \big| \Gamma^{(m_{j,j})}(\mu) - \Gamma^{(m_{j+1,j+1})}(\mu) \big| \geq 2^{1-j} \big\rbrace
\end{align}
%Writing $L\in \mathbb{Z}^+$ to be such that $\mu \in \mathcal{U}_L  / \mathcal{U}_{L-1}$, and let $a \geq L$ be such that
Thanks to a union-of-events bound, for any $N \geq 1$, and using the bound in Lemma \ref{Lemma discrete m n approximation},
\begin{align}
\sup_{\mathbf{z}_0,\mathbf{g}_0 \in \mathcal{Y}^N}  Q^N_{\mathbf{z}_0,\mathbf{g}_0} \bigg( \hat{\mu}^N \in \bigcup_{j \geq k} \mathcal{A}_j \bigg) \leq \sum_{j=k}^{\infty}\exp\big( - N 2^j \big).
\end{align}
It thus follows from the Borel-Cantelli Lemma that there exists a random $k$ such that $\hat{\mu}^N \notin \mathcal{A}_j$ for all $j\geq k$, and so the limit in \eqref{eq: empirical measure limit} exists (almost surely). \eqref{eq: uniform approximation RD Derivative} follows analogously.

\end{proof}

\begin{lemma} \label{Lemma discrete m n approximation}
(i) $\Gamma^{(m)}: \mathcal{U} \to \mathbb{R}$ is continuous. (ii) Moreover, for any $a,j \in \mathbb{Z}^+$, there exists $m_{a,j}$ such that for all $m\geq m_{a,j}$ and all $n\geq m$,
\begin{align}
\sup_{\mathbf{z}_0,\mathbf{g}_0 \in \mathcal{Y}^N} N^{-1}\log Q^N_{\mathbf{z}_0,\mathbf{g}_0} \big( \big| \Gamma^{(m)}\big( \hat{\mu}^N(\mathbf{y},\tilde{\mathbf{G}}) \big) -  \Gamma^{(n)}\big( \hat{\mu}^N(\mathbf{y},\tilde{\mathbf{G}}) \big) \geq 2^{-j} \big) \leq - 2^{a}.
\end{align}
\end{lemma}
\begin{proof}
(i) The continuity of $\Gamma^{(m)}$ is almost immediate from the definition.

(ii) For any $t\in [0,T]$, write $t^{(m)} = \sup\lbrace t^{(m)}_b \; : t^{(m)}_b \leq t \rbrace$. Starting with the discrete approximation to the stochastic integral, we can thus write
\begin{align}
\sum_{b=0}^{m-1}\sigma_{t^{(m)}_a}^{-2} \big(  G^p_{t_b^{(m)}} - \tau^{-1} z^{p}_{t_b^{(m)}}\big) \big( z^{p}_{t^{(m)}_{b+1}} -  z^{p}_{t^{(m)}_{b}} \big) 
= \int_0^T \sigma^{-2}_{t^{(m)}}\big(   G^{p}_{t^{(m)}} - \tau^{-1} z^p_{t^{(m)}} \big) dz^{p}_{t}.
\end{align}
Hence,
\begin{multline}
 \sum_{b=0}^{m-1} \mathbb{E}^{\mu}\bigg[ \sigma_{t^{(m)}_b}^{-2} \big(  G^p_{t_b^{(m)}} - \tau^{-1} z^{p}_{t_b^{(m)}}\big) \big( z^{p}_{t^{(m)}_{b+1}} -  z^{p}_{t^{(m)}_{a}}   \big) \bigg] - \\
  \sum_{ab=0}^{n-1} \mathbb{E}^{\mu}\bigg[ \sigma_{t^{(n)}_b}^{-2} \big(  G^p_{t_b^{(n)}} - \tau^{-1} z^{p}_{t_b^{(n)}}\big) \big( z^{p}_{t^{(n)}_{b+1}} -  z^{p}_{t^{(n)}_{b}}   \big) \bigg] \\
  = \mathbb{E}^{\mu}\bigg[  \int_0^T\bigg\lbrace \sigma^{-2}_{t^{(m)}}\big(   G^{p}_{t^{(m)}} - \tau^{-1} z^p_{t^{(m)}} \big)- \sigma^{-2}_{t^{(n)}}\big(   G^{p}_{t^{(n)}} - \tau^{-1} z^p_{t^{(n)}} \big)  \bigg\rbrace dz^{p}_{t} \bigg] \\
  =  \mathbb{E}^{\mu}\bigg[  \int_0^T  \sum_{p\in I_M} (f^p_{t^{(m)}} - f^p_{t^{(n)}}) dz_t^p \bigg]
 \end{multline}
 where $ f^p_t = \sigma_t^{-2} \big(   G^p_t - \tau^{-1} z^p_t \big)$. Writing
 \begin{align}
 f^{p,j}_t = \sigma_t^{-2} \big(   G^{p,j}_t - \tau^{-1} z^{p,j}_t \big),
 \end{align}
 we obtain that
 \begin{align}
  \mathbb{E}^{\hat{\mu}^N}\bigg[  \int_0^T  \sum_{p\in I_M} (f^p_{t^{(m)}} - f^p_{t^{(n)}}) dz_t^p \bigg] = N^{-1}\sum_{j\in I_N \fatsemi p\in I_M} \int_0^T \big(f^{p,j}_{t^{(m)}} - f^{p,j}_{t^{(n)}} \big)dy^{p,j}_t.
  \end{align}
  The quadratic variation of this stochastic integral is
  \begin{align}
  (QV)^{(m,n),N}_t = N^{-2}\sum_{j\in I_N \fatsemi p\in I_M} \int_0^t\big(f^{p,j}_{s^{(m)}} - f^{p,j}_{s^{(n)}} \big)^2 \sigma_s^2 ds 
  \end{align}
  By definition of the set $\mathcal{U}_a$, if $\hat{\mu}^N \in \mathcal{U}_a$, then for any $\delta > 0$, one can find $m_{\delta}$ such that as long as $m,n \geq m_{\delta}$, necessarily
  \[
 (QV)^{(m,n),N}_T \leq N^{-1} \delta.
  \]
  Then writing $w(\cdot)$ to be a standard Brownian Motion, using the Dambin-Dubins-Schwarz \cite{Karatzas1991} time-rescaled representation of the stochastic integral, as long as $(m,n) \geq m_{\delta}$,
 \begin{align}
 \mathbb{P}\bigg(  \hat{\mu}^N \in \mathcal{U}_a \; , \; \bigg|  \int_0^T  \sum_{p\in I_M} (f^p_{t^{(m)}} - f^p_{t^{(n)}}) dz_t^p \bigg| \geq   \frac{\epsilon}{2} \bigg)
 \leq &\mathbb{P}\big( \big| w\big( N^{-1} \delta \big) \big| \geq \epsilon \big) \\ =& \exp\big( - N \epsilon^2 / (8\delta) \big) 
 \leq \exp( - NL ),
 \end{align}
as long as we choose $\delta$ sufficiently small.

The other terms in
\[ 
 \Gamma^{(m)}\big( \hat{\mu}^N(\mathbf{y},\tilde{\mathbf{G}}) \big) -  \Gamma^{(n)}\big( \hat{\mu}^N(\mathbf{y},\tilde{\mathbf{G}}) \big) 
 \]
 are treated similarly (observe that they are just Riemann Sums, so its straightforward to control their difference from the limiting integral). 
  
\end{proof}

We now prove Theorems \ref{Theorem Coupled Large Deviation Principle} and \ref{Theorem Coupled Large Deviation Principle Unconditioned}.
\begin{proof}
In the case of connectivity-independent initial conditions (Case 2 of the Assumptions), the theorem follows from Corollary \ref{Corollary Uncoupled Large Deviation Principle Unconditioned}. Since the relative entropy is only zero when its two arguments are identical, any zero must be a fixed point of the operator $\Phi$. It is proved in the following Lemma that there is a unique zero.

For the rest of this proof, we prove the theorem in the case of connectivity-dependent initial conditions. We start by proving that for any $\epsilon > 0$, there must exist a measure $\mu \in \mathcal{U}$ such that
\begin{align}
\lsup{N}N^{-1} \log \mathbb{P}\big( d_W(\hat{\mu}^N(\mathbf{z},\mathbf{G}) , \mu) \leq \epsilon \big) = 0. \label{eq: proposition end corollary}
\end{align}
Write $\mathfrak{U} = \mathcal{U}_a$, where $a$ is large enough that
\[
\lsup{N} \sup_{(\mathbf{z}_0,\mathbf{g}_0) \in \mathcal{Y}^N}N^{-1}\log Q^N_{\mathbf{z}_0,\mathbf{g}_0}\big( \hat{\mu}^N(\mathbf{y},\tilde{\mathbf{G}}) \in \mathcal{U}_a \big) < - C.
\]
where $C$ is the upperbound for $\Gamma$ in Lemma \ref{Lemma Boundedness of Gamma}. This is possible thanks to the Exponential Tightness. Thanks to the Radon-Nikodym derivative identity in \eqref{eq: Radon-Nikodym N system}, we thus find that
%It is a consequence of the Large Deviation Principle in Theorem \ref{Theorem Uncoupled Large Deviation Principle} that the system is exponentially tight (see the discussion on page 8 of \cite{Dembo1998}). Accordingly, there must exist a compact set $\mathfrak{U} \subset \mathcal{U}$ such that
\begin{align}
\lsup{N} N^{-1} \log \mathbb{P}\big( \hat{\mu}^N(\mathbf{z},\tilde{\mathbf{G}}) \notin \mathfrak{U} \big) < 0.
\end{align}
Thus for \eqref{eq: proposition end corollary} to hold, it suffices that we prove that there exists $\mu \in \mathfrak{U}$ such that
\begin{align}
\lsup{N}N^{-1} \log \mathbb{P}\big(\hat{\mu}^N(\mathbf{z},\mathbf{G})  \in \mathfrak{U} \; , \; d_W(\hat{\mu}^N(\mathbf{z},\mathbf{G}) , \mu) \leq \epsilon \big) = 0. \label{eq: proposition end corollary 2}
\end{align}
Since $\mathfrak{U}$ is compact, for any $\epsilon > 0$, we can obtain a finite covering of $\mathfrak{U}$ of the form
\begin{align}
\mathfrak{U} \subseteq \bigcup_{i=1}^{\mathcal{N}_{\epsilon}} B_{\epsilon}(\mu_i),
\end{align}
where $\mu_i \in \mathfrak{U}$. By a union of events bound,
\begin{align}
 0 = & \lim_{N\to\infty} N^{-1} \log \mathbb{P}\big( \hat{\mu}^N(\mathbf{z},\tilde{\mathbf{G}}) \in \mathfrak{U} \big) \\
 \leq & \max_{1\leq i \leq \mathcal{N}_{\epsilon}} \bigg\lbrace  \lsup{N} N^{-1} \log \mathbb{P}\big( \hat{\mu}^N(\mathbf{z},\tilde{\mathbf{G}}) \in B_{\epsilon}(\mu_i) \big) \bigg\rbrace \label{eq: max over many}
\end{align}
If our proposition in \eqref{eq: proposition end corollary 2} were to be false, then \eqref{eq: max over many} would be strictly negative, which would be a contradiction.

Write $\mu_{(k)} \in \mathfrak{U}$ to be such that 
\begin{align}
\lsup{N} N^{-1} \log \mathbb{P}\big( d_W(\hat{\mu}^N(\mathbf{z},\mathbf{G}) ,  \mu_{(k)}) \geq k^{-1} \big)  = 0.
\end{align}
Let $\mu \in \mathfrak{U}$ be any measure such that for some subsequence $(p_k)_{k\geq 1}$, $\lim_{k\to\infty}\mu_{(p_k)} = \mu$ (this must be possible because $\mathfrak{U}$ is compact).

We next claim that
\begin{align}
\lim_{\epsilon \to 0^+} \linf{N} N^{-1}\log  \mathbb{P}\big( d_W\big(\hat{\mu}^N(\mathbf{z},\mathbf{G}) ,  \mu \big)  \leq \epsilon \big) = -   \mathcal{I}(\mu) 
+ \Gamma(\mu)   \label{eq: to shown J mu minimum}
\end{align}
Indeed writing $\mathcal{A}_{\epsilon} = \big\lbrace  d_W\big(\hat{\mu}^N(\mathbf{z},\mathbf{G}) ,  \mu \big)  \leq \epsilon  \big\rbrace$,
\begin{align}
  \mathbb{P}\big( d_W\big(\hat{\mu}^N(\mathbf{z},\mathbf{G}) ,  \mu \big)  \leq \epsilon \big) =& \mathbb{E}^{\gamma}\bigg[ \int_{\mathbb{R}^{MN}} P^N_{\mathbf{J},\mathbf{x}}(\mathcal{A}_{\epsilon}) \rho^N_{\mathbf{J}}(\mathbf{x}) d\mathbf{x} \bigg] \\
  =& \mathbb{E}^{\gamma}\bigg[ \int_{\mathbb{R}^{MN}} \mathbb{E}^{P^N_{\mathbf{x}}}\big[ \exp\big( N \Gamma(\hat{\mu}^N) \big)  \chi\lbrace \mathcal{A}_{\epsilon}\rbrace \big]\rho^N_{\mathbf{J}}(\mathbf{x}) d\mathbf{x} \bigg] \\
    =& \int_{\mathbb{R}^{MN}}  \mathbb{E}^{\gamma}\bigg[\mathbb{E}^{P^N_{\mathbf{x}}}\big[ \exp\big( N \Gamma(\hat{\mu}^N) \big)  \chi\lbrace \mathcal{A}_{\epsilon}\rbrace \big]\rho^N_{\mathbf{J}}(\mathbf{x}) \bigg] d\mathbf{x}  \\
       =& \int_{\mathbb{R}^{MN}}  \mathbb{E}^{\gamma}\bigg[ \mathbb{E}^{\gamma}\bigg[\mathbb{E}^{P^N_{\mathbf{x}}}\big[ \exp\big( N \Gamma(\hat{\mu}^N) \big)  \chi\lbrace \mathcal{A}_{\epsilon}\rbrace \big]\rho^N_{\mathbf{J}}(\mathbf{x}) \; \bigg| \; \mathbf{G}_0 \bigg] \bigg]d\mathbf{x}  
\end{align}
and in this last step we first perform the conditional expectation, for $\gamma$ conditioned on the values of $\lbrace G^{p,j}_0 \rbrace_{j\in I_N \fatsemi p\in I_M}$.

Now, recall that
\[
\rho^N_{\mathbf{J}}(\mathbf{z}_0) = (Z^N_{\mathbf{J}})^{-1} \chi\big\lbrace \hat{\mu}^N(\mathbf{z}_0,\mathbf{G}_0) \in B_{\delta_N}(\kappa) \big\rbrace .
\]
Furthermore, writing
\[
u_N =  N^{-1}\log \mathbb{E}[Z^N_{\mathbf{J}}],
\]
our assumption on the initial condition dictates that for any $\delta > 0$,
\begin{equation}
 \lsup{N} N^{-1} \log \mathbb{P}\big( \big| N^{-1}\log Z^N_{\mathbf{J}} - u_N \big| \geq \delta \big) < 0.
\end{equation}
Next, we claim that
\begin{align}
\lim_{\epsilon \to 0^+} \inf_{\nu \in \mathfrak{U} \cap \mathcal{A}_{\epsilon}}\Gamma(\nu) = \Gamma(\mu).\label{eq: Approximate Gamma by Continuous Functions}
\end{align}
Indeed \eqref{eq: Approximate Gamma by Continuous Functions} is a consequence of Lemma \ref{Lemma Approximate Gamma by Continuous Functions}: this Lemma implies that $\Gamma$ can be approximated arbitrarily well by continuous functions over $\mathfrak{U}$.

We thus obtain that 
\begin{align}
&\lim_{\epsilon \to 0^+} \linf{N} \inf_{(\mathbf{z}_0,\mathbf{G}_0)}N^{-1}\log  \big( d_W\big(\hat{\mu}^N (\mathbf{z},\mathbf{G}) ,  \mu \big)  \leq \epsilon \big) \\ =& \Gamma(\mu) +  \lim_{\epsilon \to 0^+} \linf{N} \big\lbrace- u_N + N^{-1} \log \int_{\mathbb{R}^{MN}}  \mathbb{E}^{\gamma}\big[ Q^N_{\mathbf{x},\mathbf{G}_0}\big( \mathcal{A}_{\epsilon} \big) \big] \chi \big\lbrace \hat{\mu}^N(\mathbf{z}_0, \mathbf{G}_0) \in B_{\delta_N}(\kappa) \big\rbrace d\mathbf{x} \big\rbrace\\
=&   \Gamma(\mu) - \lim_{\epsilon \to 0^+}\inf_{\nu \in \mathcal{A}_{(\epsilon)}} \mathcal{I}(\nu),
\end{align}
since (by definition)
\[
N^{-1} \log \int_{\mathbb{R}^{MN}}  \mathbb{E}^{\gamma}\big[ \chi \big\lbrace \hat{\mu}^N(\mathbf{z}_0, \mathbf{G}_0) \in B_{\delta_N}(\kappa) \big] d\mathbf{z}_0 = u_N,
\]
and we have employed the uniform lower bound in \eqref{LDP lower bound}. The lower semi-continuity of $\mathcal{I}$ implies that
\[
\lim_{\epsilon \to 0^+}\inf_{\nu \in \mathcal{A}_{(\epsilon)}} \mathcal{I}(\nu) = \mathcal{I}(\mu).
\]
We thus obtain \eqref{eq: to shown J mu minimum}, as required.

Next, we must show that $\mathcal{J}(\mu) = \mathcal{I}(\mu) - \Gamma(\mu)$ (recall the definition of $ \mathcal{J}(\mu)$ in \eqref{eq: J rate function}. Now
\begin{align}
\frac{dS_{\mu,\mathbf{z}_0,\mathbf{g}_0}}{dQ_{\mu,\mathbf{z}_0,\mathbf{g}_0}}\bigg|_{\mathcal{F}_T} = \exp\bigg( \sum_{p\in I_M} \int_0^T \sigma_s^{-2}  \big(g^{p}_s- \tau^{-1} z^p_s   \big)dz^{p}_s -\frac{1}{2}\sigma_s^{-2}  \big(g^{p}_s- \tau^{-1}z^p_s   \big)^2 ds \bigg)
\end{align}
Substituting this identity into the proposed rate function definition in \eqref{eq: J rate function},
\begin{align}
 \mathbb{E}^{\kappa}\bigg[ \mathcal{R}\big(\mu_{\mathbf{z}_0,\mathbf{g}_0} || S_{\mu,\mathbf{z}_0,\mathbf{g}_0} \big) \bigg] = \mathcal{I}(\mu) - \Gamma(\mu),
 \end{align}
as required.

The above reasoning dictates that there must be at least one $\xi$ such that $\mathcal{J}(\xi) = 0$. In fact, it is proved in Lemma \ref{Lemma at most one Fixed Point} that there can only be one measure $\xi$ such that $\mathcal{J}(\xi) = 0$. Furthermore it follows from Lemma \ref{Lemma at most one Fixed Point}, over small enough time increments, the mapping $\Phi_t$ must be a contraction. This implies \eqref{eq: convergence xi n approximations}.
\end{proof}

\begin{lemma}\label{Lemma Boundedness of Gamma}
There exists a constant $C > 0$ such that
\begin{align}
\lsup{N}N^{-1}\log \mathbb{P}\big(  \Gamma^N_{\mathbf{J},T}(\mathbf{z}) \geq C  \big) < 0.
\end{align}
\end{lemma}
\begin{proof}
For any $\ell > 0$,
\begin{multline}
\lsup{N}N^{-1}\log \mathbb{P}\bigg(  \Gamma^N_{\mathbf{J},T}(\mathbf{z}) \geq C  \bigg) \leq
\max\bigg\lbrace \lsup{N}N^{-1}\log \mathbb{P}\big( \norm{\mathcal{J}_N} > \ell  \big)  , \\  \lsup{N}N^{-1}\log \mathbb{P}\big( \norm{\mathcal{J}_N} \leq \ell ,  \Gamma^N_{\mathbf{J},T}(\mathbf{z}) \geq C  \big)  \bigg\rbrace
\end{multline}
Thanks to Lemma \ref{Lemma bound on operator norm}, $ \lsup{N}N^{-1}\log \mathbb{P}\big( \norm{\mathcal{J}_N} > \ell  \big) $ converges to $-\infty$ as $\ell \to \infty$. It thus suffices that we prove that, for abitrary $\ell >0$, there exists $C_{\ell} > 0$ such that
\begin{align}
 \lsup{N}N^{-1}\log \mathbb{P}\big( \norm{\mathcal{J}_N} \leq \ell ,  \Gamma^N_{\mathbf{J},T}(\mathbf{z}) \geq C_{\ell}  \big)  < 0.
\end{align}
Now, leaving out the negative-semi-definite terms, we find that 
\begin{align}
\Gamma^N_{\mathbf{J},T}(\mathbf{z}) \leq N^{-1}\sum_{j\in I_N \fatsemi p\in I_M} \int_0^T \sigma_s^{-2}  \big(\tilde{G}^{p,j}_s- \tau^{-1} y^{p,j}_s   \big)dy^{p,j}_s
\end{align}
Furthermore, writing $h^p_s = \sigma_s^{-2}  \big(\tilde{G}^{p,j}_s- \tau^{-1} y^{p,j}_s   \big)$, and assuming that $ \norm{\mathcal{J}_N} \leq \ell$, one finds that 
\begin{align}
\sum_{j\in I_N}(h^{p,j}_s)^2 \leq & 2 \sigma_s^{-4} \sum_{j\in I_N} \big\lbrace (\tilde{G}^{p,j}_s)^2 + \tau^{-2} (y^{p,j}_s)^2 \big\rbrace \\
\leq & 2 \sigma_s^{-4} \sum_{j\in I_N} \big\lbrace \ell \lambda(y^{p,j}_s)^2+ \tau^{-2} (y^{p,j}_s)^2 \big\rbrace \\
\leq & 2 \sigma_s^{-4} \sum_{j\in I_N} \big\lbrace \ell C_{\lambda}^2 (y^{p,j}_s)^2+ \tau^{-2} (y^{p,j}_s)^2 \big\rbrace .
\end{align}
We thus find that, thanks to Lemma , for any $L>0$ there exists a constant $\bar{C}_L > 0$ such that
\begin{align}
\lsup{N} N^{-1}\log \mathbb{P}\big(  \sup_{p\in I_M}\sum_{j\in I_N}(h^{p,j}_s)^2 \geq N \bar{C}_L \big) \leq -L.
\end{align}
Write
\[
\mathcal{H}_N = \bigg\lbrace  \sup_{p\in I_M}\sum_{j\in I_N}(h^{p,j}_s)^2 \leq N \bar{C}_L \bigg\rbrace .
\]
We thus find that, 
\begin{multline}
\lsup{N} N^{-1}\log \mathbb{P}\big( \norm{\mathcal{J}_N} \leq \ell \;, \; \mathcal{H}_N \; ,\; \Gamma^N_{\mathbf{J},T}(\mathbf{z}) \geq C_{\ell}  \big) \\
\leq \max\bigg\lbrace \lsup{N} N^{-1}\log \mathbb{P}\big( \mathcal{H}_N^c \big), \\
\lsup{N} N^{-1}\log \mathbb{P}\big( \norm{\mathcal{J}_N} \leq \ell \;, \; \mathcal{H}_N \; ,\; \Gamma^N_{\mathbf{J},T}(\mathbf{z}) \geq C_{\ell}  \big) \bigg\rbrace \\
\leq  \max\bigg\lbrace -L , 
\lsup{N} N^{-1}\log \mathbb{P}\big( \norm{\mathcal{J}_N} \leq \ell \;, \; \mathcal{H}_N \; ,\; \Gamma^N_{\mathbf{J},T}(\mathbf{z}) \geq C_{\ell}  \big) \bigg\rbrace
\end{multline}
Furthermore, using the Dambins-Dubins Schwarz Theorem \cite{Karatzas1991}, and writing $w(t)$ to be 1D Brownian Motion,
\begin{multline}
\lsup{N} N^{-1}\log \mathbb{P}\big( \norm{\mathcal{J}_N} \leq \ell \;, \; \mathcal{H}_N \; ,\; \Gamma^N_{\mathbf{J},T}(\mathbf{z}) \geq C_{\ell}  \big)  \\
\leq \lsup{N} N^{-1}\log \mathbb{P}\big( \sup_{s\in [0,T]} \big| w\big( \bar{C}N^{-1} s \big) \big| \geq C_{\ell} \big) \leq - L,
\end{multline}
as long as $C_{\ell}$ is sufficiently large, using standard properties of Brownian Motion.
\end{proof}

We now prove that the rate function $\mathcal{J}$ has a unique minimizer (i.e. we prove Lemma \ref{Lemma unique minimizer rate function}).
\begin{lemma}\label{Lemma at most one Fixed Point}
There exists a unique fixed point $\xi$ of $\Phi$ in $\mathcal{U}$. Furthermore $\xi$ is such that for any $\mu \in \mathcal{U}$, writing $\xi_{(1)} = \mu$ and $\xi_{(n+1)} = \Phi(\xi_{(n)})$, it holds that
\begin{align}
\xi = \lim_{n\to\infty} \xi_{(n)} \label{convergence xi (n)}
\end{align}
\end{lemma}
\begin{proof}
We start by considering the following restricted map $\tilde{\Phi}: \mathfrak{Q} \to \mathfrak{Q}$. This is such that $\tilde{\Phi}(\mu) = \nu^{(1)}$, where $\nu = \Phi(\alpha)$ for any $\alpha \in \mathcal{U}$ such that $\alpha^{(1)} = \mu$. Define $d^{(2)}_t: \mathfrak{Q} \times \mathfrak{Q} \to \mathbb{R}^+$ analogously.

We are going to demonstrate that there is a constant $c > 0$ such that for all $\mu,\nu \in \mathfrak{Q}$,
\begin{align}
d^{(2)}_t \big( \tilde{\Phi}_t(\mu), \tilde{\Phi}_t(\nu) \big) \leq c \sqrt{t} d^{(2)}_t(\mu,\nu).
\end{align}
For any $\mu,\nu \in  \mathfrak{Q}$, we construct a particular $\zeta$ that is within $\eta \ll 1$ of realizing the infimum in the definition of the Wasserstein distance in \eqref{eq: zeta Wasserstein L squares def}. To do this, we employ the construction of Lemma \ref{beta mapping continuous}. Let $\mathbf{G}^{\mu} , \mathbf{G}^{\nu} $ be  $\mathcal{C}([0,T],\mathbb{R}^{M})$-valued random variables (in the same probability space), with joint probability law $\beta_{\mu,\nu}$. Then for Brownian motions $\big( W^p_t \big)_{p\in I_M}$, define
\begin{align}
dz^{\nu,p}_t &= \big( - \tau^{-1} z^{\nu,p}_t + G^{\nu,p}_t \big) dt + \sigma_t dW^p_t  \\
dz^{\mu,p}_t &= \big( - \tau^{-1} z^{\mu,p}_t + G^{\mu,p}_t \big) dt + \sigma_t dW^p_t .
\end{align}
The initial conditions are identical: $z^{\nu,p}_0 = z^{\mu,p}_0$. We immediately see that
\begin{align}
\frac{d}{dt}\big( z^{\nu,p}_t - z^{\mu,p}_t \big) = - \tau^{-1} \big( z^{\nu,p}_t - z^{\mu,p}_t \big) + G^{\nu,p}_t - G^{\mu,p}_t,
\end{align}
and hence
\begin{align}
\frac{d}{dt}\big( z^{\nu,p}_t - z^{\mu,p}_t \big)^2 &= -2 \tau^{-1} \big( z^{\nu,p}_t - z^{\mu,p}_t \big)^2 +2 \big( z^{\nu,p}_t - z^{\mu,p}_t \big) \big(G^{\nu,p}_t - G^{\mu,p}_t\big) \text{ and thus }\\
\big( z^{\nu,p}_t - z^{\mu,p}_t \big)^2 &= \int_0^t \exp\big(2 (s-t) / \tau \big) 2 \big( z^{\nu,p}_s - z^{\mu,p}_s \big) \big(G^{\nu,p}_s - G^{\mu,p}_s\big) ds.
\end{align}
It follows from this that there exists a constant $c > 0$ such that for all $t\in [0,T]$,
%\[
%\sup_{t\in [0,s]}\big| z^{\nu,p}_t - z^{\mu,p}_t \big|^2 \leq c s \sup_{t\in [0,s]}\big| G^{\nu,p}_t - G^{\mu,p}_t \big|^2
%\]
%Therefore, there exists a constant such that 
\begin{align}
d^{(2)}_t \big( \tilde{\Phi}_t(\mu) , \tilde{\Phi}_t(\nu) \big) &\leq  c t d^{(2)}_t\big( \beta_{\mu},\beta_{\nu}\big) \\
&\leq c C_{\lambda} t  d^{(2)}_t(\mu,\nu), \label{Lemma bound tilde Phi}
\end{align}
using Lemma \ref{beta mapping continuous}. Thus for small enough $t$, there is a unique fixed point of $\tilde{\Phi}_t$ (the mapping upto time $t$). Iterating this argument, we find a unique fixed point for $\tilde{\Phi}$. The uniqueness for $\tilde{\Phi}$ in turn implies uniqueness for $\Phi$, thanks to the identity in Lemma \ref{beta mapping continuous}. 

To see why \eqref{convergence xi (n)} holds. First consider arbitrary $\nu_{(1)} \in \mathfrak{Q}$, and define $\nu_{(n+1)} = \tilde{\Phi}(\nu_{(n)})$. The above bound in \eqref{Lemma bound tilde Phi} implies that necessarily $\big(\nu_{(n)} \big)_{n\geq 1}$ is Cauchy. It then immediate follows that for any $\xi_{(1)} \in \mathcal{U}$ with first marginal equal to $\nu_{(1)}$, and writing $\xi_{(n+1)} = \Phi(\xi_{(n)})$, it must be that $\big(\xi_{(n)} \big)_{n\geq 1}$ is Cauchy.

Finally we note that $d^{(2)}$ metrizes weak convergence, thanks to Lemma \ref{Lemma Equivalent Metric}.

\end{proof}

\appendix
\section{Bounding Fluctuations of the Noise}
For the processes $(\mathbf{y}^j_{[0,T]})_{j\in I_N}$ that are defined in \eqref{eq: y dynamics}, define the empirical measure
\begin{align}
\hat{\mu}^N(\mathbf{y}) = N^{-1}\sum_{j\in I_N} \delta_{\mathbf{y}^j_{[0,T]}} \in  \mathcal{P}\big( \mathcal{C}([0,T], \mathbb{R}^M) \big) .
\end{align}
Next, we bound the probability of the empirical being in the set $\mathfrak{Q}_a$, defined in \eqref{eq: set Q L}, which we recall
\begin{multline} \label{eq: set Q L 2}
\mathcal{Q}_{\mathfrak{a}} = \bigg\lbrace \mu \in \mathcal{P}\big( \mathcal{C}([0,T], \mathbb{R}^M) \big) \; : \;  \sup_{m \geq \mathfrak{a}} \sup_{0\leq i \leq m}\mathbb{E}^{\mu}\big[ \sup_{M \in I_M} (w^p_{t^{(m)}_{i+1}} - w^p_{t^{(m)}_i})^2\big] > \Delta_m^{1/4} \text{ and } \\
\mu \in \mathcal{K}_{\mathfrak{a}} \text{ and }\sup_{p\in I_M} \mathbb{E}^{\mu}[ \sup_{t\in [0,T]}(y^p_t)^2 \big] \leq \mathfrak{a} \bigg\rbrace
\end{multline}
where $\Delta_m = T / m$ and $t^{(m)}_i = iT/m$. The main result of this section is the following.
\begin{lemma} \label{Lemma Brownians Exponentially Tight}
For any $L > 0$, there exists $\mathfrak{a} \in \mathbb{Z}^+$ such that for all $N\geq 1$, %the set $\mathcal{Q}_L \subset \mathcal{P}\big( \mathcal{C}([0,T],\mathbb{R}^M) \big)$ (defined in \eqref{eq: set Q L} below) is compact. Furthermore for any $q > 0$, there exists $L_q$ such that for every $N \geq 1$,
\begin{align}
\sup_{(\mathbf{z}_0,\mathbf{g}_0)} N^{-1} \log P^N_{\mathbf{z}_0}\big( \hat{\mu}^N(\mathbf{y}) \notin \mathcal{Q}_{\mathfrak{a}} \big) \leq - L.
\end{align}
\end{lemma}
\begin{proof}

Employing a union-of-events bound, or any $(\mathbf{z}_0,\mathbf{g}_0)\in \mathcal{Y}^N$,
\begin{multline}
N^{-1} \log P^N_{\mathbf{z}_0}\big( \hat{\mu}^N(\mathbf{y}) \notin \mathcal{Q}_{\mathfrak{a}} \big) \leq N^{-1}\log \bigg\lbrace P^N_{\mathbf{z}_0}\bigg( \sup_{p\in I_M} N^{-1}\sum_{j\in I_N} \sup_{t\in [0,T]}(y^{p,j}_t)^2  > \mathfrak{a} \bigg) \\ + 
P^N_{\mathbf{z}_0}\bigg( \sup_{0\leq t \leq \Delta_m} \sum_{j\in I_N}\sup_{0 \leq i \leq m-1} \sup_{p\in I_M} \big| y^{p,j}_{t + t^{(m)}_{i}} - y^{p,j}_{t^{(m)}_i} \big|^2 \geq Na \Delta_m \bigg)
 +  P^N_{\mathbf{z}_0}\big( \hat{\mu}^N(\mathbf{y}) \notin \mathcal{K}_{\mathfrak{a}} \big)  
 \bigg\rbrace . \label{eq: threeway decomposition}
 \end{multline}

 With a view to bounding the first term on the RHS, since $y^{p,j}_0 = z^{p,j}_0$,
 \[
(y^{p,j}_t)^2 \leq 2 \big( y^{p,j}_t - y^{p,j}_0 \big)^2 + 2 (z^{p,j}_0)^2 .
 \]
Thus for a positive constant $b > 0$,
\begin{multline}
\mathbb{E}^{P^N_{\mathbf{z}_0}}\bigg[ \exp\bigg( b\sup_{p\in I_M}  \sum_{j\in I_N} \sup_{t\in [0,T]}(y^{p,j}_t)^2    \bigg)\bigg] \\
\leq \mathbb{E}^{P^N_{\mathbf{z}_0}}\bigg[ \exp\bigg( 2b\sup_{p\in I_M} \sum_{j\in I_N}  (z^{p,j}_0)^2 +  2b\sum_{j\in I_N} \sup_{t\in [0,T]}(y^{p,j}_t -z^{p,j}_0)^2    \bigg)\bigg].
\end{multline}
Thus, thanks to Chernoff's Inequality,
\begin{align}
N^{-1}\log P^N_{\mathbf{z}_0}\bigg( \sup_{p\in I_M} N^{-1}\sum_{j\in I_N} \sup_{t\in [0,T]}(y^{p,j}_t)^2  > \mathfrak{a} \bigg)
\leq  \frac{2b}{N} \sup_{p\in I_M} \sum_{j\in I_N}  (z^{p,j}_0)^2 \\+ N^{-1}\log  \mathbb{E}^{P^N_{\mathbf{z}_0}}\bigg[\exp\bigg( 2b\sum_{j\in I_N} \sup_{t\in [0,T]}(y^{p,j}_t -z^{p,j}_0)^2    \bigg)\bigg] - b\mathfrak{a}.
\end{align}
The first term on the RHS is bounded for all $N$ and all $(\mathbf{z}_0,\mathbf{g}_0) \in \mathcal{Y}^N$. For the second term on the RHS, standard theory on stochastic processes implies that the exponential moment exists, as long as $b$ is small enough. Thus, taking
 $\mathfrak{a} \to \infty$, the RHS can be made arbitrarily small. We thus find that
 \begin{align}
 \lim_{\mathfrak{a}\to\infty} \sup_{N\geq 1}\sup_{(\mathbf{z}_0,\mathbf{g}_0) \in \mathcal{Y}^N} N^{-1} \log P^N_{\mathbf{z}_0}\bigg( \sup_{p\in I_M} N^{-1}\sum_{j\in I_N} \sup_{t\in [0,T]}(y^{p,j}_t)^2  > \mathfrak{a} \bigg) = -\infty . \label{eq: bound supremum temporary}
 \end{align}
The Lemma now follows from applying \eqref{eq: bound supremum temporary}, Lemma \ref{Lemma bound fluctuations of Brownian} and Lemma \ref{Lemma compact K L} to \eqref{eq: threeway decomposition}.
% 
%thanks to Lemma \ref{Lemma bound fluctuations of Brownian},
%\begin{align}
%\sum_{m \geq \mathfrak{a}} \mathbb{P}\big( \hat{\mu}^N(\mathbf{W}) \notin \mathcal{Q}_{\mathfrak{a}} \big) \leq  \sum_{m \geq \mathfrak{a}}m^N\exp\bigg( N \mathfrak{C} + N - m^{1/4}N / 4 \bigg).
%\end{align}
%Now
%\begin{multline}
% \sum_{m \geq \mathfrak{a}}m^N\exp\bigg( N \mathfrak{C} + N - m^{1/4}N / 4 \bigg) \\= \exp\bigg( - \mathfrak{a}^{1/8}N / 4 \bigg) \sum_{m \geq \mathfrak{a}}m^N\exp\bigg( N \mathfrak{C} + N - m^{1/4}N / 4 + \mathfrak{a}^{1/8}N / 4 \bigg)
%\end{multline}
%and one easily checks that
%\[
%\lim_{\mathfrak{a} \to \infty} \sup_{N\geq 1}\sum_{m \geq \mathfrak{a}}m^N\exp\bigg( N \mathfrak{C} + N - m^{1/4}N / 4 + \mathfrak{a}^{1/8}N / 4 \bigg) = 0.
%\]
%Thus for large enough $\mathfrak{a}$, the Lemma must hold.
\end{proof}

The following result is well-known. Nevertheless we sketch a quick proof for clarity.
\begin{lemma} \label{Lemma compact K L}
For any $L > 0$, there exists a compact set $\mathcal{K}_{L}$ such that for all $N \geq 1$,
\begin{align}
\sup_{(\mathbf{z}_0,\mathbf{g}_0) \in \mathcal{Y}^N}N^{-1}\log P^N_{\mathbf{z}_0}\big( \hat{\mu}^N(\mathbf{y}) \notin \mathcal{K}_L \big) \leq - L
\end{align}
\end{lemma}
\begin{proof}
The following property follows straightforwardly from properties of the stochastic integral (noting that the diffusion coefficient is uniformly bounded): for any $\epsilon > 0$, there exists a compact set $\mathcal{C}_{\epsilon} \subset \mathcal{C}([0,T],\mathbb{R}^M)$ such that for all $j \in I_N$ such that $\| z_0^j \| \leq \epsilon^{-1}$,
\begin{align}
\sup_{j\in I_N}P^N_{\mathbf{z}_0}\big( y^j_{[0,T]} \notin \mathcal{C}_{\epsilon} \big) \leq \epsilon.
\end{align}
Write
\begin{align}
u^N_{\epsilon} = \sup_{(\mathbf{z}_0,\mathbf{g}_0) \in \mathcal{Y}^N}N^{-1}\sum_{j\in I_N} \chi\lbrace \| \mathbf{y}^j_{0} \| \geq \epsilon^{-1} \rbrace ,
\end{align}
and note that our assumptions on $\mathcal{Y}^N$ dictates that
\begin{align}
\lim_{\epsilon \to 0^+} \lim_{N\to\infty} u_{\epsilon}^N = 0.
\end{align}
For any $ m\in \mathbb{Z}^+$, define the set $\mathcal{L}_{m,\delta} \subset \mathcal{P}\big( \mathcal{C}([0,T],\mathbb{R}^M) \big)$ to be such that
\begin{align}
\mathcal{L}_{m,\delta} = \big\lbrace \mu \in   \mathcal{P}\big( \mathcal{C}([0,T],\mathbb{R}^M) \big) \; : \; \mu(\mathcal{C}_{m^{-1}}) \geq \delta \big) \big\rbrace 
\end{align}
We claim that for any $m\geq 1$, there exists $\delta_m > 0$ such that
\begin{align} \label{eq: exponential inequality claim}
\sup_{N\geq 1} \sup_{(\mathbf{z}_0,\mathbf{g}_0) \in \mathcal{Y}^N} N^{-1} \log P^N_{\mathbf{z}_0}\big( \hat{\mu}^N\big(\mathbf{y}\big) \notin  \mathcal{L}_{m,\delta_m} \big) \leq - m
\end{align}
To see this, employing a Chernoff Inequality, for a constant $b > 0$, for any $(\mathbf{z}_0,\mathbf{g}_0) \in \mathcal{Y}^N$,
\begin{align}
N^{-1} \log P^N_{\mathbf{z}_0}\big( \hat{\mu}^N\big(\mathbf{y}\big) \notin  \mathcal{L}_{m,\delta} \big) \leq &
\mathbb{E}^{P^N_{\mathbf{z}_0}}\bigg[ \exp\bigg( b\sum_{j\in I_N}\chi\lbrace \mathbf{y}^j_{[0,T]} \notin \mathcal{C}_{m^{-1}} \rbrace - Nb \delta \bigg) \bigg] \\
\leq& - b\delta + N^{-1} \log \big\lbrace (\epsilon + u^N_{\epsilon})\big( \exp(b) -1 \big) +1  \big\rbrace^N \\
=&  - b\delta + \log \big\lbrace(\epsilon + u^N_{\epsilon})\big( \exp(b) -1 \big) +1 \big\rbrace
\end{align}
Taking $\epsilon$ to be sufficiently small, and $b$ sufficiently large, we obtain \eqref{eq: exponential inequality claim}.

Now, for an integer $m_L$ to be specified further below, define $\mathcal{K}_L = \bigcap_{m \geq m_L}  \mathcal{L}_{m,\delta_m}$. Prokhorov's Theorem implies that $\mathcal{K}_L$ is compact. Employing a union-of-events bound, we obtain that
\begin{align}
P^N_{\mathbf{z}_0}\big( \hat{\mu}^N(\mathbf{y}) \notin \mathcal{K}_L \big) &\leq \sum_{m\geq m_L} \exp( - mN) \\
 &\leq \exp(-m_LN) \sup_{n \geq 1}\sum_{j=0}^{\infty} \exp( - jN).
\end{align}
We thus find that, for large enough $m_L$,
\[
\sup_{N\geq 1} N^{-1}\log \mathbb{P}\big( \hat{\mu}^N(\mathbf{y}) \notin \mathcal{K}_L \big) \leq - L,
\]
as required.
\end{proof}

\begin{lemma}\label{Lemma bound fluctuations of Brownian}
There exists a constant $\mathfrak{C}$ such that for any positive integer $m$ and any $a > 0$, writing $\Delta_m = Tm^{-1}$ and $t^{(m)}_i = Ti/m$, for any $N\geq 1$,
\begin{align}
\sup_{(\mathbf{z}_0,\mathbf{g}_0) \in \mathcal{Y}^N} N^{-1} \log P^N_{\mathbf{z}_0}\bigg( \sup_{0\leq t \leq \Delta_m} \sum_{j\in I_N}\sup_{0 \leq i \leq m-1} \sup_{p\in I_M} \big| y^{p,j}_{t + t^{(m)}_{i}} - y^{p,j}_{t^{(m)}_i} \big|^2 \geq Na \Delta_m \bigg) \leq  
\mathfrak{C} + \log m - \frac{a}{4} 
\end{align}

\end{lemma}
\begin{proof}
Define, for $t \in [0, \Delta_m)$,
\[
f^N_t =   \sum_{j\in I_N}\sup_{0 \leq i \leq m-1} \sup_{p\in I_M} \big( y^{p,j}_{t + t^{(m)}_i} - y^{p,j}_{t^{(m)}_i} \big)^2 
\]
Notice that $t \to f^N_t$ is a submartingale. Thus, writing  $a=(4\Delta_m)^{-1}$, $\exp\big(a f^N_t\big)$ is a submartingale. Therefore, thanks to Doob's Submartingale Inequality,
\begin{align}
\mathbb{P}\big( f^N_t \geq N x \big) &\leq \mathbb{E}\bigg[ \exp\big( a f^N_T - a Nx \big) \bigg] \\
&\leq \big\lbrace mM\big( 1-2\Delta_m \bar{\sigma} a \big)^{-1/2} \big\rbrace^N \exp(-aNx) \\
&= \big\lbrace mM 2^{1/2} \big\rbrace^N \exp\big(-Nx / (4\Delta_m) \big) 
\end{align}
\end{proof}
\bibliographystyle{plain}
\bibliography{SGBib2.bib}

\begin{thebibliography}{10}

\bibitem{Adler2007}
Robert Adler and Jonathan Taylor.
\newblock {\em Random Fields and Geometry}, volume~53.
\newblock Springer, 2019.

\bibitem{Alaoui2020}
Ahmed~El Alaoui, Andrea Montanari, and Mark Sellke.
\newblock Optimization of mean-field spin glasses.
\newblock {\em Arxiv Preprint}, 2020.

\bibitem{BenArous1995}
G~Ben Arous and A~Guionnet.
\newblock Large deviations for langevin spin glass dynamics.
\newblock {\em Probability Theory and Related Fields}, 102, 1995.

\bibitem{BenArous2006}
Gerard~Ben Arous, Amir Dembo, and Alice Guionnet.
\newblock Cugliandolo-kurchan equations for dynamics of spin-glasses.
\newblock {\em Probability Theory and Related Fields}, 136:619--660, 2006.

\bibitem{Arous1997}
Gerard~Ben Arous and Alice Guionnet.
\newblock Symmetric langevin spin glass dynamics.
\newblock {\em The Annals of Probability}, 25:1367--1422, 1997.

\bibitem{BenArousMei2019}
Gerard~Ben Arous, Song Mei, Andrea Montanari, and Mihai Nica.
\newblock The landscape of the spiked tensor model.
\newblock {\em Communications on Pure and Applied Mathematics}, 72:2282--2330,
  11 2019.

\bibitem{Brunel2000}
Nicolas Brunel.
\newblock Dynamics of sparsely connected networks of excitatory and inhibitory
  spiking neurons.
\newblock {\em Journal of Computational Neuroscience}, 8:183--208, 2000.

\bibitem{Brunel2003}
Nicolas Brunel and Xiao-Jing Wang.
\newblock What determines the frequency of fast network oscillations with
  irregular neural discharges? i. synaptic dynamics and excitation-inhibition
  balance.
\newblock {\em Journal of Neurophysiology}, 90:415--430, 2003.

\bibitem{Budhiraja2019}
Amarjit Budhiraja and Paul Dupuis.
\newblock {\em Analysis and Approximation of Rare Events}, volume~94.
\newblock Springer, 2019.

\bibitem{Cabana2018}
Tanguy Cabana and Jonathan~D. Touboul.
\newblock Large deviations for randomly connected neural networks: I. spatially
  extended systems.
\newblock {\em Advances in Applied Probability}, 50:983--1004, 2018.

\bibitem{Cabana2018a}
Tanguy Cabana and Jonathan~D. Touboul.
\newblock Large deviations for randomly connected neural networks: Ii.
  state-dependent interactions.
\newblock {\em Advances in Applied Probability}, 50:983--1004, 2018.

\bibitem{Cessac2019}
B.~Cessac.
\newblock Linear response in neuronal networks: from neurons dynamics to
  collective response.
\newblock 5 2019.

\bibitem{Parisi2023}
Patrick Charbonneau, Enzo Marinari, Mark Mezard, Giorgio Parisi, Federico
  Ricci-Tersenghi, Gabriella Sicuro, and Francesco Zamponi, editors.
\newblock {\em Spin Glass Theory and Far Beyond}.
\newblock World Scientific, 2023.

\bibitem{Crisanti2018}
A.~Crisanti and H.~Sompolinsky.
\newblock Path integral approach to random neural networks.
\newblock {\em Physical Review E}, 98:1--20, 2018.

\bibitem{Crisanti1993}
Andrea Crisanti, Heinz Horner, and H.J. Sommers.
\newblock The spherical p-spin interaction spin-glass model.
\newblock {\em Zeitschrift fur Physik B Condensed Matter}, 92:257--271, 1993.

\bibitem{Cugliandolo1993}
L.~F. Cugliandolo and J.~Kurchan.
\newblock Analytical solution of the off-equilibrium dynamics of a long-range
  spin-glass model.
\newblock {\em Physical Review Letters}, 71:173--176, 1993.

\bibitem{Dembo2020}
Amir Dembo and Eliran Subag.
\newblock Dynamics for spherical spin glasses: Disorder dependent initial
  conditions.
\newblock {\em Journal of Statistical Physics}, 181:465--514, 2020.

\bibitem{Dembo1998}
Amir Dembo and Ofer Zeitouni.
\newblock {\em Large Deviations Techniques and Applications 2nd Edition}.
\newblock Springer, 1998.

\bibitem{Fasoli2019}
Diego Fasoli and Stefano Panzeri.
\newblock Stationary-state statistics of a binary neural network model with
  quenched disorder.
\newblock {\em Entropy}, pages 1--30, 2019.

\bibitem{Faugeras2015}
Olivier Faugeras and James MacLaurin.
\newblock Asymptotic description of neural networks with correlated synaptic
  weights.
\newblock {\em Entropy}, 17:4701--4743, 2015.

\bibitem{Faugeras2019a}
Olivier Faugeras, Emilie Soret, and Etienne Tanré.
\newblock Asymptotic behaviour of a network of neurons with random linear
  interactions.
\newblock {\em Preprint HAL Id : hal-01986927}, 2019.

\bibitem{Gamarnik2021}
David Gamarnik.
\newblock The overlap gap property: A topological barrier to optimizing over
  random structures.
\newblock {\em Proceedings of the National Academy of Scientists}, 2021.

\bibitem{Grunwald1996}
M~Grunwald.
\newblock Sanov results for glauber spin-glass dynamics.
\newblock {\em Probability Theory and Related Fields}, 106:187--232, 1996.

\bibitem{Guionnet1997a}
A~Guionnet.
\newblock Averaged and quenched propagation of chaos for spin glass dynamics.
\newblock {\em Probability Theory and Related Fields}, 109:183--215, 1997.

\bibitem{Guionnet1997}
Alice Guionnet and Boguslaw Zegarlinski.
\newblock Decay to equilibrium in random spin systems on a lattice.
\newblock {\em Journal of Statistical Physics}, 86:899--904, 1997.

\bibitem{Helias2020}
Moritz Helias and David Dahmen.
\newblock {\em Statistical Field Theory for Neural Networks}.
\newblock Springer, 2020.

\bibitem{Karatzas1991}
Ioannis Karatzas and Steven Shreve.
\newblock Brownian motion and stochastic calculus 2nd edition, 1991.

\bibitem{Landau2018}
Itamar~Daniel Landau and Haim Sompolinsky.
\newblock Coherent chaos in a recurrent neural network with structured
  connectivity.
\newblock {\em PLoS Computational Biology}, 14:1--27, 2018.

\bibitem{Lindgren2013}
George Lindgren, Holger Rootzen, and Maria Sandsten.
\newblock {\em Stationary Stochastic Processes for Scientists and Engineers}.
\newblock Chapman Hall, 2013.

\bibitem{Lucon2017}
Eric Lucon.
\newblock Quenched large deviations for interacting diffusions in random media.
\newblock {\em Journal of Statistical Physics}, 166:1405--1440, 2017.

\bibitem{Moynot2002}
Olivier Moynot and Manuel Samuelides.
\newblock Large deviations and mean-field theory for asymmetric random
  recurrent neural networks.
\newblock {\em Probability Theory and Related Fields}, 123:41--75, 5 2002.

\bibitem{Ocker2017}
Gabriel~Koch Ocker, Yu~Hu, Michael~A. Buice, Brent Doiron, Kresimir Josic,
  Robert Rosenbaum, and Eric Shea-Brown.
\newblock From the statistics of connectivity to the statistics of spike times
  in neuronal networks.
\newblock {\em Current Opinion in Neurobiology}, 46:109--119, 2017.

\bibitem{Rosenbaum2014}
Robert Rosenbaum and Brent Doiron.
\newblock Balanced networks of spiking neurons with spatially dependent
  recurrent connections.
\newblock {\em Physical Review X}, pages 1--9, 2014.

\bibitem{Rosenbaum2017}
Robert Rosenbaum, Matthew~A. Smith, Adam Kohn, Jonathan~E. Rubin, and Brent
  Doiron.
\newblock The spatial structure of correlated neuronal variability.
\newblock {\em Nature Neuroscience}, 20:107--114, 2017.

\bibitem{Segadlo2022}
Kai Segadlo, Bastian Epping, Alexander~Van Meegen, David Dahmen, Michael
  Krämer, and Moritz Helias.
\newblock Unified field theoretical approach to deep and recurrent neuronal
  networks.
\newblock {\em Journal of Statistical Mechanics: Theory and Experiment}, 2022,
  10 2022.

\bibitem{Sompolinsky1988}
H.~Sompolinsky, A.~Crisanti, and H.~J. Sommers.
\newblock Chaos in random neural networks.
\newblock {\em Physical Review Letters}, 61:259--262, 1988.

\end{thebibliography}

\end{document}